\documentclass{article}

\usepackage[preprint]{neurips_2026}

\usepackage[utf8]{inputenc} 
\usepackage[T1]{fontenc}    
\usepackage{hyperref}       
\usepackage{url}            
\usepackage{booktabs}       
\usepackage{amsfonts}       
\usepackage{nicefrac}       
\usepackage{microtype}      
\usepackage{xcolor}         
\usepackage{amsmath}
\usepackage{multirow}
\usepackage{graphicx}
\usepackage{algorithm}
\usepackage{algorithmic}
\usepackage{amsthm}
\usepackage{amssymb}
\usepackage{bbding}
\usepackage{stfloats}

\newtheorem{theorem}{Theorem}

\newtheorem{corollary}[theorem]{Corollary}
\newtheorem{lemma}[theorem]{Lemma}
\newtheorem{definition}{Definition}
\newtheorem{remark}{Remark}
\newtheorem{assumption}{Assumption}

\newcommand{\E}{\mathbb{E}}

\newcommand{\Cov}{\mathrm{Cov}}
\newcommand{\rank}{\mathrm{rank}}
\newcommand{\F}{\mathcal{F}}

\newcommand{\KL}{D}

\title{FedAttr: Towards Privacy-preserving Client-Level Attribution in Federated LLM Fine-tuning}

\author{%
  Su Zhang \\
  Department of Computer Science\\
  University of Maryland, College Park\\
  \texttt{suzhang1@umd.edu} \\
  \And
  Junfeng Guo \\
  Department of Computer Science\\
  University of Maryland, College Park\\
  \texttt{gjf2023@umd.edu} \\
  \And
  Heng Huang \\
  Department of Computer Science\\
  University of Maryland, College Park\\
  \texttt{heng@umd.edu} \\
}

\begin{document}

\maketitle

\begin{abstract}
Watermark radioactivity testing type of methods can detect whether a model was trained on watermarked documents, and have become key tools for protecting data ownership in the fine-tuning of large language models (LLMs). Existing works have proved their effectiveness in centralized LLM fine-tuning. However, this type of method faces several challenges and remains underexplored in federated learning (FL), a widely-applied paradigm for fine-tuning LLMs collaboratively on private data across different users. FL mainly ensures privacy through secure aggregation (SA), which allows the server to aggregate updates while keeping clients' updates private. This mechanism preserves privacy but makes it difficult to identify
which client trained on watermarked documents. In this work, we propose \textbf{FedAttr}, a new \textit{client-level attribution} protocol for FL. FedAttr identifies which clients trained on watermarked data via a paired-subset-difference mechanism, while preserving the privacy guarantees of SA and FL performance. FedAttr proceeds in three steps: (i) estimate each client's update by differencing two SA queries, (ii) score the estimate with the watermark detector via differential scoring,
and (iii) combine scores across rounds via Stouffer method. We theoretically show that FedAttr produces an unbiased estimator of each client's update with bounded mutual information leakage (\textit{i.e.,} $O(d^*/N)$ per-round update). Moreover, FedAttr empirically achieves 100\% TPR and 0\% FPR, outperforming all baselines by at least 44.4\% in TPR or 19.1\% in FPR, with only
6.3\% overhead relative to FL training time. Ablation studies confirm that FedAttr is robust to protocol parameters and configurations.
\end{abstract}


\section{Introduction}\label{Introduction}



Large language models (LLMs) are increasingly fine-tuned on documents obtained from external sources, often under license terms that restrict data use to the licensee's own training. Watermark radioactivity testing type of methods have emerged as a practical tool for detecting violations of such terms: the data provider embeds a watermark into the documents, and a watermark detection test can later determine whether a model was trained on watermarked documents~\citep{sander2024radioactive,cui2025fictitious}. This approach has been validated for centralized LLM fine-tuning.

In federated learning (FL), where multiple institutions jointly fine-tune a model on their private data without sharing it~\citep {ye2024openfedllm,fan2023fatellm}, however, the watermark radioactivity test faces two challenges and remains underexplored. A global model radioactivity test can still detect that the trained model was influenced by watermarked documents. However, it cannot identify which client used them. This distinction matters because license terms are often granted to individual institutions: the data provider needs to know exactly which institutions violated the terms. Identifying the responsible clients is challenging for two reasons. First, FL systems widely adopt the secure aggregation (SA) mechanism to protect client privacy, which hides each client's individual update from the server. This mechanism preserves privacy but makes it difficult to identify 
which client trained on watermarked documents. Second, even if individual updates were available, the global model already carries watermark signals from previous rounds, so even a benign client's update appears watermarked when tested, giving over 57\% FPR in our experiments. Existing FL forensic methods~\citep{zhang2022fldetector,jia2024tracing} do not resolve these challenges: they are designed to detect \emph{adversarial} 
poisoning attacks, in which malicious clients 
send manipulated updates. In our setting, all clients faithfully follow the 
protocol, producing none of the adversarial 
signals these methods rely on. Moreover, both 
methods require plaintext access to updates, violating SA.

In this work, we propose \textbf{FedAttr}, a novel \textit{client-level attribution protocol} for federated LLM fine-tuning. FedAttr preserves the standard FL training and SA protocol. It identifies which clients were trained on watermarked data via the paired-subset-difference mechanism, while preserving the privacy guarantees of SA and FL performance. 
FedAttr proceeds in three steps. First, to overcome SA's restriction on 
observing individual updates, FedAttr estimates each client's update by differencing two authorized SA subset queries, one that includes the target client and one that excludes it, yielding 
an unbiased estimator with bounded variance 
(proved by Theorems~\ref{thm:unbiased}--\ref{thm:variance}). Second, since the global model accumulates 
watermark bias across rounds, 
FedAttr reduces this bias by scoring each estimate \emph{relative} 
to the current global model. Third, since the per-round watermark signal is weak, FedAttr combines per-round differential scores 
across $T$ communication rounds via Stouffer's 
method to identify which client is watermarked. The score gaps between watermarked and benign clients grow with 
$\sqrt{T}$, driving error rates to zero 
exponentially (proved by Theorem~\ref{thm:stouffer}).

We summarize our contributions as follows.

\noindent\textbf{(1) Problem and protocol.} We formalize 
\textit{client-level attribution} problem and propose \textbf{FedAttr}, which combines a client-level update estimator, differential scoring, and cross-round Stouffer combination to decide which client uses the watermarked documents through the SA mechanism.

\noindent\textbf{(2) Theoretical guarantees.} We prove the client-level update estimator is unbiased with bounded variance 
(Theorems~\ref{thm:unbiased},~\ref{thm:variance}), and derive two-sided 
exponential error bounds for cross-round Stouffer combination that drive 
false negatives to zero in $T$ 
(Theorem~\ref{thm:stouffer}).

\noindent\textbf{(3) Privacy analysis.} We bound the per-round mutual 
information leakage of the estimator about each client's update by 
$O(d^*/N)$, where $d^*$ is the effective subspace dimension and $N$ is 
the subset size of SA queries (Theorem~\ref{thm:diffuse-query}).

\noindent\textbf{(4) Empirical validation.} In federated LoRA fine-tuning experiment, FedAttr achieves 100\% TPR at 0\% FPR within $5$ rounds across two watermark families and two aggregation 
strategies, outperforming all baselines by at least 44.4\% in TPR or 19.1\% in FPR, with only 
6.3\% overhead relative to FL training time. Ablation studies confirm that FedAttr is robust to parameters and configurations.

\begin{figure*}[t]
\centering
\includegraphics[width=\textwidth, trim=0 10 0 10, clip]{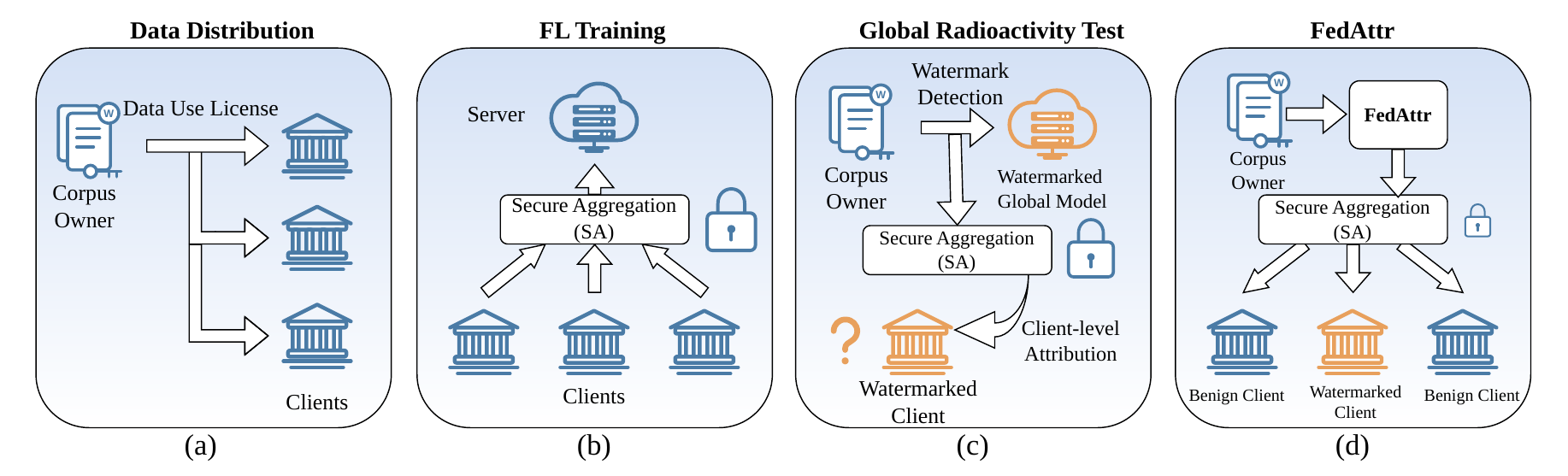}
\vspace{-14pt}
\caption{%
Overview of Data Attribution in Federated Learning.
\textbf{(a)}~The corpus owner distributes 
watermarked documents to clients under a data-use license.
\textbf{(b)}~Clients collaboratively fine-tune a 
shared model via federated learning through secure 
aggregation (SA).
\textbf{(c)}~A global radioactivity test can 
detect the watermark signal of the trained model, but cannot identify which clients are responsible without violating SA.
\textbf{(d)}~FedAttr identifies which clients use the watermark documents through SA.
}
\label{fig:overview}
\vspace{-7.62mm}
\end{figure*}
\section{Related Work}\label{sec:related}

\noindent\textbf{Federated LLM fine-tuning.}
Federated learning~\citep{mcmahan2017communication} 
has been extended to LLM fine-tuning via 
parameter-efficient adapters such as 
LoRA~\citep{hu2022lora}, with aggregation 
strategies including 
FedIT~\citep{zhang2024fedit} and 
FLoRA~\citep{wang2024flora}. We evaluate both 
in our experiments. Detailed descriptions of two aggregation strategies are in Appendix~\ref{app:exper}.

\noindent\textbf{Training-data watermarking.}
Embedding detectable signals into training data 
so that downstream models inherit measurable 
traces, known as 
\emph{radioactivity}~\citep{sablayrolles2020radioactive}, was 
extended to LLMs in two forms: 
\citet{sander2024radioactive} generate documents 
with a watermarked LLM whose green-token bias 
transfers to models fine-tuned on them, 
and~\citet{cui2025fictitious} inject fabricated 
entity-attribute pairs that the fine-tuned model 
memorizes and can be detected via QA probes. These methods assume a single 
training party; in FL, the global model aggregates 
all clients' updates, so a detection on the global model no longer identifies which client used the 
watermarked data. FedAttr reuses these detectors 
as black-box scoring functions and resolves 
this attribution problem.

\noindent\textbf{Federated forensics.}
FLDetector~\citep{zhang2022fldetector} and 
FLForensics~\citep{jia2024tracing} trace 
malicious clients in poisoning attacks by 
detecting update inconsistency and 
misclassification influence, respectively. 
Both require plaintext access to the individual 
updates and rely on adversarial signals absent 
in our non-adversarial setting, where clients 
faithfully follow the FL protocol.

\noindent\textbf{Secure aggregation and privacy.}
SA protocols~\citep{bonawitz2017practical} 
enable the server to compute subset sums of 
client updates without observing the individual 
updates.~\citet{elkordy2023howmuch} 
provide the first mutual information bound on 
per-round leakage under standard SA protocol. FedAttr's privacy analysis 
extends this framework to the multi-query 
setting required by client-level attribution.
\section{Problem setup}
\label{sec:setup}
 
\subsection{Problem Formulation}\label{sec:fl-setup}

\textbf{Federated fine-tuning system.} FL enables $K$ clients $\{1, \ldots, K\}$ to collaboratively train a shared global LLM over $T$ communication rounds under the coordination of a central server. At each round $t$, the server distributes the current global model parameters $w^{t-1} \in \mathbb{R}^d$ to all clients; client $i$ locally fine-tune $w^{t-1}$ on its private dataset $\mathcal{D}_i$ and returns the resulting parameter update $\Delta_i^t \in \mathbb{R}^d$. Then the server aggregates updates into the new global model according to a federated aggregation rule~\citep{mcmahan2017communication}:
\vspace{-2mm}
\begin{equation}
w^t = w^{t-1} + \sum
_{i=1}^K p_i \Delta_i^t,
\label{eq:agg-rule}
\end{equation}where aggregation weights $p_i \geq 0$ satisfy $\sum_{i=1}^K p_i = 1$, typically $p_i = |\mathcal{D}_i| / \sum_j |\mathcal{D}_j|$. We assume all $K$ clients participate in every communication round.

\noindent\textbf{Secure aggregation.} Secure aggregation (SA) aims to preserve client updates' privacy in FL systems~\citep{bonawitz2017practical}. SA allows the server to compute the sum of client updates $S^t(W)$ over any subset of clients $W \subseteq [K]$ with $|W| \geq N_{\mathrm{sa}}$, while keeping each individual update hidden:
\vspace{-1mm}
\begin{equation}
S^t(W) = \sum_{j \in W} \Delta_j^t.
\vspace{-2mm}
\label{eq:sa-interface}
\end{equation}
The threshold $N_{\mathrm{sa}}$ prevents individual updates from being exposed.\footnote{SA can be implemented via multi-party 
computation~\citep{bonawitz2017practical} or homomorphic 
encryption~\citep{zhang2020batchcrypt}; FedAttr is agnostic to the 
choice of instantiation.}

\textbf{Client-level attribution problem.} An unknown set of clients trained on watermarked documents in violation of license terms, i.e., use the watermarked documents. Given access to SA aggregations $\{S^t(W)\}$ over admissible subsets $W$ and rounds 
$t \in \{1, \ldots, T\}$, the \textit{client-level attribution} problem is to output a per-client binary decision 
$r_i \in \{0, 1\}$ for each client $i \in [K]$, where $r_i = 1$ iff client $i$'s dataset $\mathcal{D}_i$ contains watermarked documents. 

\subsection{Watermark Families}
\label{sec:watermark}

FedAttr requires only a scoring function 
$\textsc{Score}(w;\,\mathcal{P})$ that returns a larger value when $w$ has been trained on watermarked data, where $w$ is the model under test and $\mathcal{P}$ is a set of evaluation prompts.

\textbf{KGW watermark}~\citep{kirchenbauer2023watermark}. Before 
distributing the documents, the data provider rephrases them with a 
watermarked LLM that partitions the vocabulary into green and red lists 
via a pseudorandom function and boosts green-token logits by $\delta$ 
during decoding. A model fine-tuned on these documents inherits the 
green-token bias~\citep{sander2024radioactive}. To detect this bias, a 
z-test compares the observed green-token ratio against the expected 
null rate $\gamma = |G|/|\mathcal{V}|$.

\textbf{Fictitious knowledge watermark}~\citep{cui2025fictitious}. The 
data provider injects fabricated entity-attribute tuples (e.g., 
``\textit{Arlo Vance was born in 1987}'') into the documents. A model 
fine-tuned on these documents memorizes the fictitious attributes. To detect memorization, each attribute is queried via QA, and per-attribute results are aggregated via Fisher's method. 

\subsection{Threat Model and Considered Scenarios}\label{sec:threat}


FedAttr targets a non-adversarial license-violation setting in which all parties are honest-but-curious: they execute the protocol faithfully but may attempt to infer private information from observation. An 
unknown subset of clients trains on watermarked documents in violation of license terms.

We consider a FL system with three parties: \emph{clients} $\{1,\ldots,K\}$ with private datasets, a \emph{server} that coordinates training via the SA interface, and a \emph{corpus owner} that holds the watermark detection key. During training, the server coordinates FL training with clients and observes only authorized subset sums $S^t(W)$ through SA. The corpus owner is not involved. After training, the 
server sends FedAttr estimates to the corpus owner, who applies the detection key and identifies which clients use the watermarked data. Neither party sees the other's private inputs: the server never learns the detection key, and the corpus 
owner never observes individual updates.

\section{FedAttr Protocol}\label{sec:protocol}

FedAttr preserves the standard FL training and SA protocol, and enables the corpus owner to identify which clients use the watermarked documents. Specifically, FedAttr estimates the update via a paired-subset-difference mechanism motivated and supported by Theorems~\ref{thm:unbiased}-~\ref{thm:variance}. FedAttr contains three steps: (i) the server estimates each client's update 
from paired SA queries and sends the estimate to the corpus owner, (ii) the corpus owner scores each estimate by the 
watermark detector, and (iii) the corpus owner combines per-round scores 
across rounds to identify the client via Stouffer method. For each stage, we perform theoretical analyses to illustrate that FedAttr can preserve the utility of client update estimates, demonstrating the effectiveness of our protocol. \textit{Algorithm~\ref{alg:fedattr}} summarizes it.

\subsection{Client-level Update Estimator}
\label{sec:surrogate}

Our goal here is to construct an unbiased estimator of any single client's update through the SA interface. The challenge is that SA hides clients' individual updates. The first step is based on a key observation: Two subset SA queries differing only in whether they include a target client $i$ must differ only by client $i$'s update in expectation. We call it the paired-subset-difference mechanism. As such, we can construct an unbiased estimator of any single client's update via paired subset SA queries. 


\noindent\textbf{Constructing the unbiased update estimator via paired subset SA queries:} Given a target client $i \in [K]$ and $N \in [K-1]$, the number of non-target clients per query.\footnote{Both subset sizes $N$ and $N+1$ must satisfy the SA protocol's authorization threshold $N_{\mathrm{sa}}$~\citep{bonawitz2017practical}; we assume this throughout. In practice $N_{\mathrm{sa}}$ is small relative to $K$.} We define two sampling families over the non-target clients $[K]\setminus\{i\}$, distinguished by whether they include the target client~$i$:
\begin{equation}\label{eq:UV}
\mathcal{U}_i^N \;=\; \bigl\{\, U \subseteq [K] : i \in U,\; |U| = N{+}1 \,\bigr\},
\qquad
\mathcal{V}_i^N \;=\; \bigl\{\, V \subseteq [K] : i \notin V,\; |V| = N \,\bigr\}.
\end{equation}
The server draws $M$ include-target subsets $U_1^t, \dots, U_{M}^t$ i.i.d.\ uniformly from $\mathcal{U}_i^N$ and $M$ exclude-target subsets $V_1^t, \dots, V_{M}^t$ i.i.d.\ uniformly from $\mathcal{V}_i^N$, then forms the round-$t$ \emph{update estimator}
\begin{equation}\label{eq:estimator}
\widehat{\Delta}_i^{\,t}
\;:=\; \frac{1}{M}\sum_{m=1}^{M} S_t(U_m^t)
\;-\; \frac{1}{M}\sum_{m=1}^{M} S_t(V_m^t).
\end{equation}
For convenience, we denote the paired
queries at round $t$ to target client $i$ by $\mathcal{Q}_i^t:=
(U_{1}^t,\ldots,U_{M}^t;
  V_{1}^t,\ldots,V_{M}^t)$, and all queries in round $t$ by $\mathcal{Q}^t=\{\mathcal Q^t_i\}_{i=1}^K$. For each non-target client 
$j \neq i$, the \emph{masking coefficient} is
$
\alpha_j^t := \frac{1}{M}\sum_{m=1}^{M} 
\mathbf{1}\{j \in U_m^t\} - \frac{1}{M}
\sum_{m=1}^{M} \mathbf{1}\{j \in V_m^t\}.
$ Note that the queries process is independent with the global model and client updates.

\noindent\textbf{Rejecting estimator when privacy condition fails:} To ensure that the non-target updates provide 
sufficient masking noise for the privacy analysis 
(Section~\ref{sec:privacy}), for instance, to 
exclude the degenerate case where the include and 
exclude subsets draw identical non-target clients, 
leaving $\widehat{\Delta}_i^{\,t} = \Delta_i^t$, the server checks the following privacy condition before querying the SA interface:
\begin{equation}\label{eq:diffuse-check}
c^t_i := \sum_{j \neq i}(\alpha_j^t)^2 \geq aN
,\qquad
M_{\mathrm{eff},i}^t := 
\frac{(c^t_i)^2}{\sum_{j \neq i}(\alpha_j^t)^4} 
\geq aN, \qquad\text{and} \qquad N < K-1,
\end{equation}
where $a := (1-\rho)/M$ and 
$\rho := N/(K-1)<1$. If the condition fails, the server 
resamples the subsets. This rejection policy depends only on the subset choice and never on client updates and the global model. We define the acceptance event at round $t$: $\mathcal A^t_i
:=
\left\{
c^t_i \ge aN,\;
M_{\mathrm{eff},i}^t \ge aN,\;N < (K-1)
\right\}$.

By the symmetry of the accepted sampling 
distribution, each 
non-target client's updates cancel in 
expectation. We prove that FedAttr constructs an unbiased estimator (Theorem~\ref{thm:unbiased}) with variance 
controlled by the non-target updates 
(Theorem~\ref{thm:variance}).

\begin{theorem}[Unbiasedness under rejection sampling]
\label{thm:unbiased}
Given any round $t$ and target client $i \in [K]$. 
Under the subset sampling with rejection described 
above, for any deterministic updates 
$\Delta_1^t, \ldots, \Delta_K^t$,
\[
\mathbb{E}\bigl[\widehat{\Delta}_i^{\,t} 
\,\big|\, \Delta_1^{\,t},\dots,\Delta_K^{\,t},\mathcal{A}^t_i
\bigr] = \Delta_i^{\,t}.
\]
\end{theorem}

\begin{theorem}[Conditional variance under 
rejection sampling]\label{thm:variance}
Under the same setting as 
Theorem~\ref{thm:unbiased}, the conditional 
covariance satisfies
\begin{equation}\label{eq:var}
\mathrm{Cov}\bigl(\widehat{\Delta}_i^{\,t} 
\,\big|\, \Delta_1^{\,t},\dots,\Delta_K^{\,t}
,\mathcal{A}^t_i\bigr)
\preceq
\frac{1}{p_{a,i}^t}
\cdot \frac{2}{M}
\cdot \frac{N(K-1-N)}{K-2} 
\cdot \Sigma_{-i}^{\,t},
\end{equation}
where $p_{a,i}^t := \Pr_{\mathcal Q^t}(\mathcal A^t_i)$ and 
$\Sigma_{-i}^t := \tfrac{1}{K-1}\sum_{j \ne i}
(\Delta_j^t - \bar\Delta_{-i}^t)
(\Delta_j^t - \bar\Delta_{-i}^t)^\top$.
\end{theorem}


\noindent Complete proofs are in Appendix~\ref{app:unbiased}-\ref{app:variance}. The rejection check introduces negligible overhead: the threshold $aN$ equals half the expected masking strength and the rejection probability decays as $e^{-\Omega(N)}$
nearly identical to unrestricted sampling, so the expected number of redraws $1/p_{a,i}^t\to 1$ exponentially fast. Theorem~\ref{thm:variance} shows that increasing the query count $M$ reduces estimator 
noise at the cost of additional SA queries. Detailed analysis of the acceptance rate is in the Appendix~\ref{app:rejection}.

\subsection{Differential Scoring}
\label{sec:diff-scoring}

In this stage, we aim to score the client's estimate with the watermark detector. The challenge is that the global model $w^{t-1}$ already carries watermark 
signals absorbed from previous rounds, causing 
the detector to assign high scores to all 
clients, including benign ones. Applying the detector directly to the updates causes 57\% FPR (demonstrated in Table~\ref{tab:main}) even with access to plaintext updates.

FedAttr addresses this via \textit{differential scoring}: it evaluates the detector at both the global model $w^{t-1}$ and the estimate model $w^{t-1} + \widehat{\Delta}_i^t$, and calculate the difference. For each round $t$ and target client $i$:
\begin{equation}
z_i^{(t)} := \textsc{Score}(w^{t-1} + \widehat{\Delta}_i^t;\; \mathcal{P}_t) - \textsc{Score}(w^{t-1};\; \mathcal{P}_t),
\label{eq:diff-score}
\end{equation}
where $\mathcal{P}_t$ is the evaluation prompt set containing watermark pattern at round $t$. Motivated by analysis in Appendix~\ref{app:differential-scoring}, differential scoring can effectively reduce the watermark bias in the global model, and $z_i^{(t)}$ measures only the contribution of client $i$'s estimated update. As demonstrated by Theorems~\ref{thm:unbiased}--\ref{thm:variance}, $\widehat{\Delta}_i^t$ is unbiased for $\Delta_i^t$, so other clients contribute only to sampling variance. 

\subsection{Cross-round Stouffer Combination}
\label{sec:stouffer}

Our goal here is to combine per-round scores across rounds to identify the watermarked client. A single round yields a weak signal because the watermark signal is not completely learned after one round of local fine-tuning; the full watermark signal emerges gradually as training $T$ increases. Moreover, the estimator would introduce sampling noise in the first stage (Section~\ref{sec:surrogate}). FedAttr accumulates this growing signal via cross-round Stouffer's 
combination~\citep{stouffer1949american}:
\begin{equation}
Z_i = \frac{1}{\sqrt{T}} \sum_{t=1}^{T} z_i^{(t)}.
\label{eq:stouffer}
\end{equation}
The corpus owner flags client $i$ as watermarked if $Z_i > \gamma$ for a fixed threshold $\gamma > 0$. 

To state the formal guarantee, we introduce a 
separation condition on the per-round scores. 
Let $\mathcal{F}_{t-1} = \sigma(w^{t-1}, \mathcal Q^1, 
\ldots, \mathcal Q^{t-1})$ denote the global model and all queries history up to round $t{-}1$.

\begin{assumption}[Watermark signal separation condition]\label{asm:separation}
There exist constants $m > 0$, $\epsilon \in [0, m)$, and $\nu > 0$ such that for each client $i \in [K]$ and round $t \in [T]$: (i)~the conditional mean $\mu_i^{(t)} := \mathbb{E}[z_i^{(t)} \mid \mathcal{F}_{t-1}]$ satisfies $\mu_i^{(t)} \geq m$ if client $i$ is watermarked and $|\mu_i^{(t)}| \leq \epsilon$ if client $i$ is benign; (ii)~the centered increment $z_i^{(t)} - \mu_i^{(t)}$ is conditionally $\nu^2$-sub-Gaussian given $\mathcal{F}_{t-1}$.
\end{assumption}

We verify this assumption empirically in Figure~\ref{fig:mechanism}(b,c). Then we introduce the following theorem:
\begin{theorem}[Stouffer error]\label{thm:stouffer}
Under Assumption~\ref{asm:separation}, for any threshold satisfying $\sqrt{T}\,\epsilon < \gamma < \sqrt{T}\,m$,
\begin{equation}\label{eq:stouffer-bounds}
\Pr\bigl(\text{error for client } i\bigr) \;\leq\;
\begin{cases}
\exp\!(-{(\gamma - \sqrt{T}\,\epsilon)^2}/{2\nu^2}) & \text{if } i \text{ is benign},\\
\exp\!(-{(\sqrt{T}\,m - \gamma)^2}/{2\nu^2}) & \text{if } i \text{ is watermarked}.
\end{cases}
\end{equation}
\end{theorem}

The proof is in Appendix~\ref{app:stouffer}. A fixed $\gamma$ controls both errors: the false-positive bound depends on $\gamma - \sqrt{T}\,\epsilon$, while the false-negative rate decays exponentially once $\sqrt{T}\,m > \gamma$. 

\section{Privacy Analysis}
\label{sec:privacy}


We analyze the information leakage of FedAttr's estimation to the corpus owner with respect to clients' updates. FedAttr operates entirely through SA, preserving the SA's privacy guarantee for each model update. However, a residual information-leakage threat remains~\citep{elkordy2023howmuch}: the server obtains a noisy estimation $\widehat{\Delta}_i^{\,t}$ of client~$i$'s actual update $\Delta_i^t$. We quantify this residual leakage using mutual information (MI), following the framework of \citet{elkordy2023howmuch}. Our analysis extends theirs from the standard SA setting to the \emph{subset-query} setting where FedAttr performs.



\noindent\textbf{Leakage metric.}
Given a target client~$i$ within round~$t$, the per-round leakage is
\begin{equation}\label{eq:leakage-metric}
I_{\mathrm{priv}}^{(t)}
:= I\!\left(\Delta_i^t;\;\widehat{\Delta}_i^{\,t}
   \;\middle|\; \mathcal{Q}^t_i,\, \mathcal{F}_{t-1}\right).
\end{equation}
This quantity measures how much information the estimator $\widehat{\Delta}_i^{\,t}$ reveals about client~$i$'s actual update, beyond what is already known from $\mathcal{F}_{t-1} = \sigma(w^{t-1}, \mathcal Q^1, 
\ldots, \mathcal Q^{t-1})$ and the round $t$ query $\mathcal{Q}^t_i$. 


Inspired by ~\citet{elkordy2023howmuch}, we propose two assumptions on the properties of the model to shed light on the leakage of MI during the update process.

\begin{assumption}[Independent under whitening]
\label{asm:bcg-regularity}
Let
\[
  Z_j^t:=(K_G^t)^{-1/2}\xi_j^t
\]
be the whitened update.  Conditioned on \(\F_{t-1}\), the
coordinates of \(Z_j^t\) are independent, centered, and have unit variance.
For every coordinate \(\ell\in[d^*]\), the scalar distribution
\(Z_{j,\ell}^t\) has finite fourth moment and finite entropic distance to the
Gaussian distribution with the same mean and variance.  More explicitly, if
\(G_\ell\sim\mathcal N(0,1)\), then there exist constants \(M_{4,\ell}<\infty\)
and \(D_{0,\ell}<\infty\), independent of \(j\), such that
\[
  \E |Z_{j,\ell}^t|^4\le M_{4,\ell},
  \qquad
  \KL(Z_{j,\ell}^t\|G_\ell)
  =
  h(G_\ell)-h(Z_{j,\ell}^t)
  \le D_{0,\ell}.
\]
These are the one-dimensional regularity conditions needed to apply the
Bobkov--Chistyakov--G{\"o}tze entropic Berry--Esseen bound used in the
independent-under-whitening case of \citet{elkordy2023howmuch}.
\end{assumption}

\begin{assumption}
\label{asm:iid-fluctuations}
The local datasets $\mathcal{D}_1, \ldots, 
\mathcal{D}_K$ are sampled i.i.d.\ from a common 
distribution, i.e., the local dataset of client~$j$ 
consists of i.i.d.\ data samples from a distribution 
$\mathcal{P}_j$, where $\mathcal{P}_j = \mathcal{P}$ 
for all $j \in [K]$. This implies that given round \(t\) and condition on the $\mathcal{F}_{t-1} = \sigma(w^{t-1}, \mathcal Q^1, 
\ldots, \mathcal Q^{t-1})$, each client update can decompose as
\[
  \Delta_j^t=\mu^t+\xi_j^t,
\]
where \(\mu^t\) is deterministic conditioned on \(\F_{t-1}\), and
\[
  \E[\xi_j^t\mid \F_{t-1}]=0,
  \qquad
  \Cov(\xi_j^t\mid \F_{t-1})=K_G^t.
\]
The \(\{\xi_j^t\}_{j=1}^K\) are conditionally i.i.d. on a common
\(d^*\)-dimensional effective subspace, where
\[
  d^*:=\rank(K_G^t).
\]
All determinants and entropies below are taken on this effective subspace.
\end{assumption}

\begin{remark}
\label{rem:regularity-role}
Assumption~\ref{asm:iid-fluctuations} is the same condition 
as \citet[Assumption~1]{elkordy2023howmuch} and ensures 
the non-target updates form an i.i.d.\ additive mask 
whose entropy can be controlled via the entropic CLT. 
The independence-under-whitening condition 
(Assumption~\ref{asm:bcg-regularity}) is satisfied when the 
stochastic gradient can be approximated by a 
distribution with independent components or by a 
multivariate Gaussian~\citep[Definition~1]{elkordy2023howmuch}.
\end{remark}

Based on two assumptions, we propose our main privacy results.

\begin{theorem}[Release-level MI leakage]
\label{thm:diffuse-query}
Suppose Assumptions~\ref{asm:bcg-regularity} and
\ref{asm:iid-fluctuations} hold.  If a query $Q_i^t$ satisfies
\[
  c_i^t\ge aN,
  \qquad
  M_{\mathrm{eff},i}^t\ge aN,\qquad \text{and} \qquad N<K-1
\]
then
\[
  I(\Delta_i^t;\widehat\Delta_i^t\mid Q_i^t,\F_{t-1})
  \le
  \frac{d^*}{2}\log\left(1+\frac1{aN}\right)
  +
  \frac{C_\xi d^*}{aN}=O\left(\frac{d^*}{N}\right).
\]

\end{theorem}

\noindent The bound has two terms: the first captures 
leakage when the masking noise is exactly Gaussian; the 
second accounts for non-Gaussianity. Compared with the single-aggregate 
bound of \citet[Theorem~1]{elkordy2023howmuch}, whose 
subset size $N{-}1$ is deterministic, FedAttr's 
effective subset size $aN$ arises from the random 
subset queries. 
The detailed 
proof is deferred to Appendix~\ref{app:privacy}.

\section{Experiments}



\subsection{Experimental Setup}

\noindent\textbf{Federated Learning Configurations.} Consistent with previous work~\citep{ye2024openfedllm,Wu2025ASO}, 
we fine-tune Llama-3.2-3B~\citep{llama3} with LoRA on UltraChat200K~\citep{ultrachat200k}, partitioned IID across $K{=}10$ clients for $T{=}5$ rounds. Default protocol parameters are $r{=}3$ watermarked clients, subset size $N{=}5$, query count $M{=}5$. We evaluate two aggregation strategies (FedIT~\citep{zhang2024fedit}, FLoRA~\citep{wang2024flora}) and two watermark families (KGW~\citep{kirchenbauer2023watermark}, Fictitious Knowledge~\citep{cui2025fictitious}). 



\noindent\textbf{Baselines.} We compare against four baselines.
(i)~\textit{Global model test}: which applies the 
watermark detector to the global model but 
cannot attribute to clients,
(ii)~\textit{Direct (oracle)}: which applies the detector to each 
client's plaintext update, violating SA,
(iii)~\textit{FLDetector}~\citep{zhang2022fldetector}, (iv)~\textit{FLForensics}~\citep{jia2024tracing}. Notably, (ii)-(iv) require 
plaintext updates and violate SA. 

\noindent\textbf{Metrics.} We report TPR (fraction 
of watermarked clients correctly flagged) and FPR 
(fraction of benign clients incorrectly flagged) for each approach. 
FedAttr flags client $i$ when its FedAttr Stouffer score $Z_i \geq 4$.
We implement each baseline following its default configurations. For each approach under different settings, we report results (\textit{i.e., mean \& std}) calculated over three random seeds. We also report the $p$-value obtained by converting 
the Stouffer statistic $Z_i$ to a one-sided standard 
normal tail probability, \textit{i.e.}, 
$p_i = 1 - \Phi(Z_i)$. In ablation studies, we additionally report $\bar{z}_{\rm pos}$ and $\bar{z}_{\rm neg}$, the mean Stouffer statistics of watermarked and benign clients, to quantify signal strength beyond the binary TPR/FPR. 
Full hyperparameters, watermark details, and baselines are in 
Appendix~\ref{app:exper}.

\subsection{Main Results}\label{sec:exp-main}

Table~\ref{tab:main} reports client-level 
attribution performance for different watermark families and FL algorithms. FedAttr achieves 100\% TPR and 0\% FPR in all 
four settings ($p < 10^{-6}$), even completely performing 
through SA . No baseline matches this 
under the same privacy constraint with SA. The direct oracle has plaintext access to each client's update but incurs FPR $\geq$ 57\%. As the global model has already learned watermark signals from earlier rounds, the detector thus assigns high scores to all clients, including benign ones. Instead, differential scoring applied in \textbf{FedAttr} significantly reduces the watermark bias of the global model, accurately isolating each client's individual effect to the watermark effectiveness.
Moreover, we observe previous forensics approach~\cite{jia2024tracing,zhang2022fldetector} performs ineffectively in our considered scenarios as the watermark signal cannot be adapted as adversary 
patterns. As a result, FLDetector 
achieves 0\% TPR, and FLForensics achieves 
at most 33--100\% TPR with 19--24\% FPR. Even more, all existing approaches require 
plaintext access to individual updates, 
violating SA, and cannot be used for watermark attribution (no $p$-value).

\begin{table}[t]
\centering
\caption{Client-level watermark attribution performance 
($T{=}5$ rounds, $\gamma{=}4.0$, 
mean$_{\pm\text{std}}$ over 3 seeds). 
FedAttr is the only method that achieves 
100\% TPR, 0\% FPR in four settings
through secure aggregation.
For FLForensics, $^{\dag}$~denotes the original implementation using HDBSCAN clustering\protect\footnotemark,
and $^{\ddag}$~denotes our adaptation using k-means clustering, since HDBSCAN fails at small $K{=}10$.}
\label{tab:main}
\small
\setlength{\tabcolsep}{3pt}
\begin{tabular}{clccccccc}
\toprule
& & & \multicolumn{3}{c}{KGW} & \multicolumn{3}{c}{Fictitious Knowledge} \\
\cmidrule(lr){4-6} \cmidrule(lr){7-9}
FL Algorithm & Baseline & SA & TPR $\uparrow$ & FPR $\downarrow$ & $p$-value & TPR $\uparrow$ & FPR $\downarrow$ & $p$-value \\
\midrule
\multirow{5}{*}{FedIT}
 & Direct (oracle)         & \XSolidBrush & $55.6_{\pm 19.3}$  & $57.1_{\pm 14.3}$  & $<10^{-3}$   & $100.0_{\pm 0.0}$ & $57.1_{\pm 14.3}$  & $<10^{-15}$   \\
 & FLDetector              & \XSolidBrush & $0.0_{\pm 0.0}$    & $19.1_{\pm 8.3}$   & ---   & $0.0_{\pm 0.0}$   & $19.1_{\pm 8.3}$   & ---   \\
 & FLForensics$^{\dag}$    & \XSolidBrush & $0.0_{\pm 0.0}$    & $0.0_{\pm 0.0}$    & ---   & $0.0_{\pm 0.0}$   & $0.0_{\pm 0.0}$    & ---   \\
 & FLForensics$^{\ddag}$   & \XSolidBrush & $33.3_{\pm 0.0}$   & $23.8_{\pm 8.3}$   & ---   & $100.0_{\pm 0.0}$ & $19.1_{\pm 8.3}$   & ---   \\
 & \textbf{FedAttr (ours)} & \Checkmark   & $\mathbf{100.0_{\pm 0.0}}$ & $\mathbf{0.0_{\pm 0.0}}$ & $<10^{-6}$ & $\mathbf{100.0_{\pm 0.0}}$ & $\mathbf{0.0_{\pm 0.0}}$ & $<10^{-29}$ \\
\midrule
\multirow{5}{*}{FLoRA}
 & Direct (oracle)         & \XSolidBrush & $100.0_{\pm 0.0}$  & $71.4_{\pm 14.3}$  & $<10^{-6}$   & $100.0_{\pm 0.0}$ & $66.7_{\pm 8.3}$   & $<10^{-20}$   \\
 & FLDetector              & \XSolidBrush & $0.0_{\pm 0.0}$    & $14.3_{\pm 0.0}$   & ---   & $0.0_{\pm 0.0}$   & $23.8_{\pm 8.3}$   & ---   \\
 & FLForensics$^{\dag}$    & \XSolidBrush & $0.0_{\pm 0.0}$    & $0.0_{\pm 0.0}$    & ---   & $0.0_{\pm 0.0}$   & $0.0_{\pm 0.0}$    & ---   \\
 & FLForensics$^{\ddag}$   & \XSolidBrush & $55.6_{\pm 19.3}$  & $23.8_{\pm 8.3}$   & ---   & $100.0_{\pm 0.0}$ & $19.1_{\pm 8.3}$   & ---   \\
 & \textbf{FedAttr (ours)} & \Checkmark   & $\mathbf{100.0_{\pm 0.0}}$ & $\mathbf{0.0_{\pm 0.0}}$ & $<10^{-10}$ & $\mathbf{100.0_{\pm 0.0}}$ & $\mathbf{0.0_{\pm 0.0}}$ & $<10^{-50}$ \\
\bottomrule
\end{tabular}
\vspace{-2pt}
\end{table}
\footnotetext{\url{https://github.com/jyqhahah/FLForensics}}
\subsection{Mechanism Analysis}\label{sec:exp-mechanism}

We further investigate the effect and soundness for each component within \textbf{FedAttr}. Figure~\ref{fig:mechanism} presents the results. Figure~\ref{fig:mechanism}(a) compares direct 
and differential scoring at round $t{=}5$. 
Direct scoring yields 57\% FPR because the 
global model has learned watermark signals
from earlier rounds. Differential scoring 
subtracts the reference score, reducing the 
accumulated bias: only watermarked clients 
retain a detectable signal.
Figure~\ref{fig:mechanism}(b) shows the 
per-round performance of differential scores. 
The watermark efficacy stays above 
$\hat{m}{=}3.3$ in each round, while the 
benign ones remain within 
$\pm\hat{\epsilon}{=}1.2$, consistent with
Assumption~\ref{asm:separation}(i)(Separation). 
Figure~\ref{fig:mechanism}(c) shows the 
empirical distribution of centered residuals 
$z_i^{(t)} - \hat{\mu}_i$, closely aligns with 
a (sub)Gaussian distribution ($\hat{\nu}{=}0.85$), supporting 
Assumption~\ref{asm:separation}(ii)(Sub-Gaussianity).
Figure~\ref{fig:mechanism}(d) shows the 
Stouffer statistic computed with varying $T'$. The watermark efficacy becomes larger than $\gamma{=}4$ 
when round $T' \geq 2$ and reaches $Z_i > 8.0$ at 
$T'{=}5$, yielding a margin of $4.0$ above 
$\gamma$. Benign clients remain near zero 
throughout. Therefore, the Stouffer process amplifies 
the signal.

\begin{figure*}[t]
\centering
\includegraphics[width=\textwidth, trim=0 0 0 0, clip]{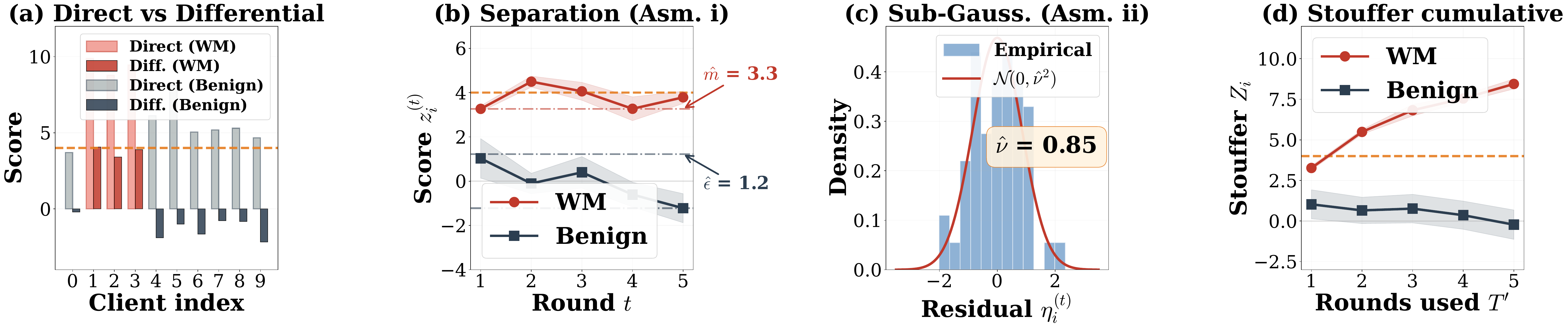}
\vspace{-18pt}
\caption{%
The effect for each 
protocol component. 
\textbf{(a)}~Differential scoring removes the accumulated 
bias. 
\textbf{(b)}~Per-round watermarked mean stays above 
$\hat{m}{=}3.3$, benign mean within 
$\pm\hat{\epsilon}{=}1.2$, validating 
Assumption~\ref{asm:separation}(i). 
\textbf{(c)}~Centered residuals $z^{(t)}_i-\hat{\mu}_i$ match a Gaussian 
($\hat{\nu}{=}0.85$), supporting 
Assumption~\ref{asm:separation}(ii). 
\textbf{(d)}~Stouffer statistic crosses $\gamma{=}4$ 
from round 2. At round 5, the statistic achieves margin $Z_i-\gamma= 4.0$.
}
\label{fig:mechanism}
\vspace{-10pt}
\end{figure*}

\vspace{-2mm}
\subsection{Ablation Studies}

We study the impact of different parameters (\textit{e.g.,} the number of watermark clients, subset size\textit{, etc}) and different configurations (\textit{e.g.,} LoRA rank, dataset, \textit{etc}); 
Figures~\ref{fig:sensitivity} (protocol parameters)
and~\ref{fig:robustness} (training configurations)
summarize the results.

\noindent\textbf{Protocol parameters 
(Figure~\ref{fig:sensitivity}).}
FedAttr achieves consistently 100\% TPR and 0\% FPR under varying amounts of watermark clients, subsets, and watermark 
ratio. The watermark efficacy $\bar{z}_{\rm pos}$ exhibits 
the U-shaped dependence on $N$ consistent with  
Theorem~\ref{thm:variance}: lowest at $N{=}5$ 
where the variance factor $N(K{-}1{-}N)/(K{-}2)$ 
peaks, and highest at $N{=}1$ where the estimator 
becomes exact 
(Figure~\ref{fig:sensitivity}(b)). Query count 
$M \geq 3$ achieves 100\% TPR and 0\% FPR.
$M{=}2$ incurs 29\% FPR due to high estimator 
noise (Figure~\ref{fig:sensitivity}(c)). The 
watermark efficacy scales roughly linearly $\textit{w.r.t.}$ watermark 
ratio, consistent with radioactivity 
theory~\citep{sander2024radioactive}.

\noindent\textbf{Robustness 
(Figure~\ref{fig:robustness}).}
\textbf{FedAttr} performs robustly under different configurations of LoRA ranks,  evaluated models, and datasets, consistently achieving 100\% TPR and 0\% FPR. 
Under severe non-IID heterogeneity 
($\alpha{=}0.1$), TPR degrades to 67\% while FPR remains 0\%; at $\alpha{=}0.05$, FPR rises to 11\%. This degradation is consistent with increased estimator variance when client 
updates diverge, and can be mitigated by 
increasing $T$ (Table~\ref{tab:noniid-recovery}). The detailed analyses are included in the Appendix.

\begin{figure*}[t]
\centering
\includegraphics[width=\textwidth, trim=0 10 0 0, clip]{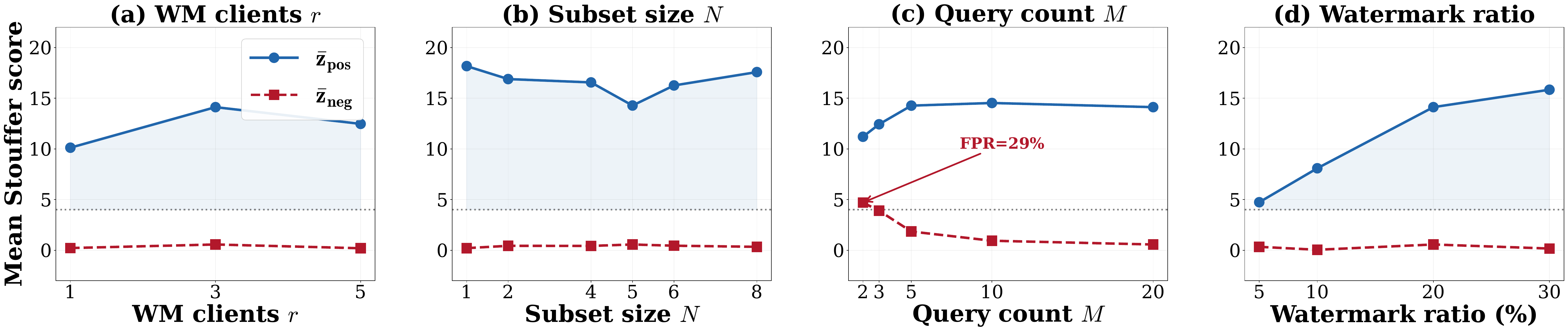}
\vspace{-15pt}
\caption{%
Protocol parameter sensitivity. 
FedAttr achieves 100\% TPR / 0\% FPR across 
all tested values except $M{=}2$, which incurs 
29\% FPR, and shows robustness under protocol parameter selection.
\textbf{(a)}~Number of watermarked clients $r$. 
\textbf{(b)}~Subset size $N$: the U-shaped curve 
validates Theorem~\ref{thm:variance}. 
(c)~Query count $M$. 
(d)~Watermark ratio: The signal scales roughly linearly with watermark 
ratio.
}
\vspace{-2mm}
\label{fig:sensitivity}
\end{figure*}

\begin{figure*}[t]
\centering
\includegraphics[width=\textwidth, trim=0 10 0 0, clip]{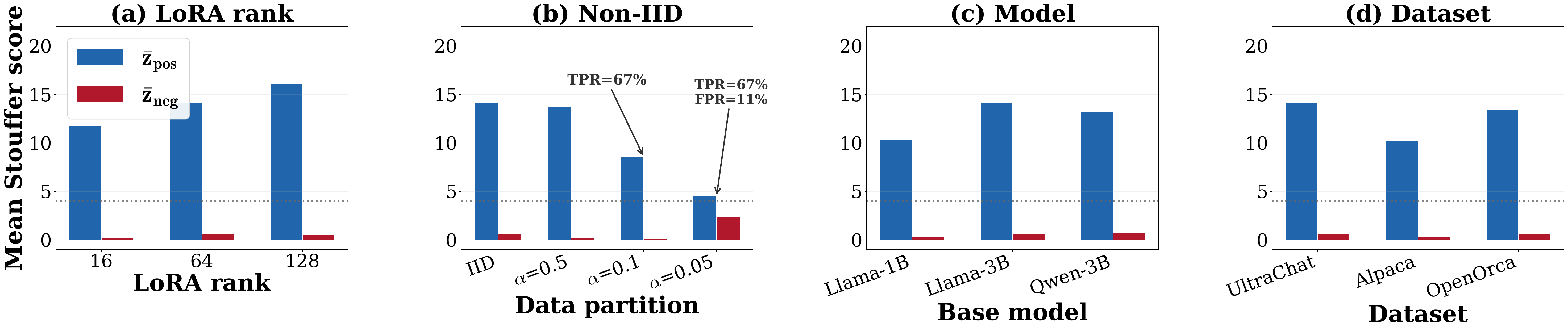}
\vspace{-13pt}
\caption{Robustness to training configurations. FedAttr achieves 100\% TPR and 0\% FPR across 
LoRA ranks, base models, and datasets. 
Attribution accuracy decreases moderately under severe non-IID partitions ($\alpha \leq 0.1$).
\textbf{(a)}~LoRA rank. 
\textbf{(b)}~Non-IID heterogeneity. (smaller Dirichlet 
$\alpha$ means more heterogeneous)
\textbf{(c)}~Base model. 
\textbf{(d)}~Training dataset.}
\label{fig:robustness}
\vspace{-10pt}
\end{figure*}

\subsection{Scalability and Overhead}
\label{sec:exp-scalability}

\noindent\textbf{Scalability.} We scale $K$ from 10 to 100. Table~\ref{tab:scalability} shows that watermark efficacy $\bar{z}_{\rm pos}$ 
decreases from 10.12 to 7.83 as $K$ increases but remains higher above threshold $\gamma$ (100\% TPR, 0\% FPR in all cases). FedAttr also performs effectively in a partial-participation setting (Table~\ref{tab:partial} in Appendix).  

\noindent\textbf{Overhead.}
FedAttr's overhead consists of SA queries 
and watermark scoring. In our main experiment, 
FedAttr issues $2MKT{=}500$ SA queries, adding 
5 minutes (1.0\%) to the 8.5-hour FL training 
time; watermark scoring adds 27 minutes (5.3\%), 
for a total overhead of 6.3\%. Both costs scale 
linearly in $K$. Since all computation runs on the server, 
it can be overlapped with clients' local 
training in the next round, effectively hiding 
the latency. Analysis of scalability and overhead is in 
Appendix~\ref{app:overhead}.

\section{Conclusion}

We introduced FedAttr, a client-level attribution protocol for federated LLM fine-tuning that identifies clients who trained on watermarked documents while preserving SA privacy. FedAttr combines unbiased update estimation from SA queries, differential scoring, and cross-round Stouffer aggregation. We provided theoretical guarantees on the estimator's unbiasedness and variance, and bounded mutual information leakage of $O(d^*/N)$ per round. Empirically, FedAttr achieves 100\% TPR, 0\% FPR, outperforming all baselines while being the only method that preserves SA privacy. 

\bibliographystyle{plainnat}
\bibliography{reference}

\newpage
\newpage
\appendix

\section{Algorithm}
\label{app:algorithm}

Algorithm~\ref{alg:fedattr} presents the FedAttr protocol. 
At each communication round $t$, clients perform standard local 
training and submit updates through secure aggregation (Lines~1--3). 
For each target client $i$, the server repeatedly samples paired 
include/exclude subsets and checks the privacy condition 
(Eq.~\eqref{eq:diffuse-check}) via rejection sampling (Lines~5--10).
Upon acceptance, the server queries the SA oracle to form the unbiased 
update estimator $\widehat{\Delta}_i^{\,t}$ (Line~11) and forwards 
it to the corpus owner, who computes the differential score 
$z_i^{(t)}$ (Lines~12--13). After all $T$ rounds, the corpus owner 
aggregates per-round scores via Stouffer's method and applies 
the threshold $\gamma$ to produce the final attribution decision 
(Lines~15--21). 


\begin{algorithm}[H]
\caption{FedAttr: Client-level Attribution through Secure Aggregation}
\label{alg:fedattr}
\begin{algorithmic}[1]
\REQUIRE Clients $[K]$, communication rounds $T$, subset size $N$, query count $M$, threshold $\gamma$, SA oracle $S_t(\cdot)$, watermark detector $\textsc{Score}(\cdot\,;\cdot)$, prompt sets $\{\mathcal{P}_t\}_{t=1}^T$
\ENSURE Attribution decision for each client $i \in [K]$
\FOR{$t = 1, \ldots, T$}

    \STATE Each client $i$ computes local update $\Delta_i^t$ and submits to the SA oracle
    \STATE Server updates global model: $w^{t} \leftarrow w^{t-1} + \sum_{i=1}^{K} p_i \Delta^t_i$

    \FOR{each target client $i \in [K]$} 

        \REPEAT
            \STATE Draw $U_1^t, \ldots, U_M^t \overset{\text{i.i.d.}}{\sim} \mathrm{Unif}(\mathcal{U}_i^N)$, {Include-target, $|U_m^t| = N{+}1$}
            \STATE Draw $V_1^t, \ldots, V_M^t \overset{\text{i.i.d.}}{\sim} \mathrm{Unif}(\mathcal{V}_i^N)$, {Exclude-target, $|V_m^t| = N$}
            \STATE Compute $\alpha_j^t \leftarrow \frac{1}{M}\sum_{m} \mathbf{1}\{j \in U_m^t\} - \frac{1}{M}\sum_{m} \mathbf{1}\{j \in V_m^t\}$ for all $j \neq i$
            \STATE Compute $c^t_i \leftarrow \sum_{j \neq i}(\alpha_j^t)^2$, \quad $M_{\mathrm{eff},i}^t \leftarrow (c^t_i)^2 \big/ \sum_{j \neq i}(\alpha_j^t)^4$
        \UNTIL{$c^t_i \geq aN$ \textbf{and} $M_{\mathrm{eff},i}^t \geq aN$ \textbf{and} $N < K{-}1$} (Eq.~\eqref{eq:diffuse-check})
        \STATE Query SA and compute: $\widehat{\Delta}_i^{\,t} \leftarrow \frac{1}{M}\sum_{m=1}^{M} S_t(U_m^t) - \frac{1}{M}\sum_{m=1}^{M} S_t(V_m^t)$
        \STATE Server sends $\widehat{\Delta}_i^{\,t}$ to corpus owner

        \STATE $z_i^{(t)} \leftarrow \textsc{Score}\bigl(w^{t-1} + \widehat{\Delta}_i^{\,t};\; \mathcal{P}_t\bigr) - \textsc{Score}\bigl(w^{t-1};\; \mathcal{P}_t\bigr)$

    \ENDFOR
\ENDFOR

\FOR{each client $i \in [K]$}
    \STATE $Z_i \leftarrow \frac{1}{\sqrt{T}} \sum_{t=1}^{T} z_i^{(t)}$
    \IF{$Z_i > \gamma$}
        \STATE Flag client $i$ as \textbf{watermarked}
    \ELSE
        \STATE Label client $i$ as \textbf{benign}
    \ENDIF
\ENDFOR
\end{algorithmic}
\end{algorithm}

\section{Proofs for Attribution Protocol}
\label{app:attribution}

\subsection{Proof of Theorem~\ref{thm:unbiased} 
(Unbiasedness under rejection sampling)}
\label{app:unbiased}

\begin{proof}
Substituting $S^t(W)=\sum_{j\in W}\Delta_j^t$ into~\eqref{eq:estimator}
gives
\[
\widehat{\Delta}_i^{\,t}
=
\frac{1}{M}\sum_{m=1}^M
\left(
\Delta_i^t + \sum_{j\in U_m^t\setminus\{i\}}\Delta_j^t
\right)
-
\frac{1}{M}\sum_{m=1}^M
\sum_{j\in V_m^t}\Delta_j^t .
\]
Since $i\in U_m^t$ for all $m$ and $i\notin V_m^t$ for all $m$, this can be
written as
\[
\widehat{\Delta}_i^{\,t}
=
\Delta_i^t
+
\sum_{j\ne i}\alpha_j^t\Delta_j^t ,
\]
where
\[
\alpha_j^t
=
\frac{1}{M}\sum_{m=1}^M
\mathbf{1}\{j\in U_m^t\setminus\{i\}\}
-
\frac{1}{M}\sum_{m=1}^M
\mathbf{1}\{j\in V_m^t\},
\qquad j\ne i .
\]
Thus it suffices to show
\[
\mathbb{E}[\alpha_j^t\mid \mathcal{A}^t_i]=0
\qquad
\text{for every } j\ne i,
\]
where $\mathcal{A}^t_i$ denotes the acceptance event
\[
\mathcal{A}^t_i
=
\{c^t\ge aN,\; M_{\mathrm{eff},i}^t\ge aN\}.
\]

For each $m$, define the non-target part of the include-target query by
\[
X_m^t := U_m^t\setminus\{i\},
\]
and define
\[
Y_m^t := V_m^t .
\]
By construction, both $X_m^t$ and $Y_m^t$ are $N$-subsets of
$[K]\setminus\{i\}$. Moreover, the proposal distribution samples
\[
X_1^t,\ldots,X_M^t,Y_1^t,\ldots,Y_M^t
\]
i.i.d. uniformly from the same family of $N$-subsets of
$[K]\setminus\{i\}$.

Let
\[
\mathcal{Q}^t
=
(X_1^t,\ldots,X_M^t;Y_1^t,\ldots,Y_M^t)
\]
denote the non-target query design. Define the swap map $T$ by
\[
T(\mathcal{Q}^t)
=
(Y_1^t,\ldots,Y_M^t;X_1^t,\ldots,X_M^t).
\]
Equivalently, after applying $T$, the corresponding include-target and
exclude-target subsets are reconstructed as
\[
(U_m^t)' = \{i\}\cup Y_m^t,
\qquad
(V_m^t)' = X_m^t .
\]
This map is well-defined because $X_m^t$ and $Y_m^t$ have the same
cardinality and both lie in $[K]\setminus\{i\}$. It is also a bijection.
Since the proposal distribution samples the $X$'s and $Y$'s i.i.d. from
the same distribution, the proposal distribution is invariant under $T$.

For every non-target client $j\ne i$,
\[
\alpha_j^t(T(\mathcal{Q}^t))
=
\frac{1}{M}\sum_{m=1}^M \mathbf{1}\{j\in Y_m^t\}
-
\frac{1}{M}\sum_{m=1}^M \mathbf{1}\{j\in X_m^t\}
=
-\alpha_j^t(\mathcal{Q}^t).
\]
Therefore the swap map sends the coefficient vector
\[
\alpha^t=(\alpha_j^t)_{j\ne i}
\]
to $-\alpha^t$.

The acceptance event $\mathcal{A}$ depends on the query design only
through
\[
c^t=\sum_{j\ne i}(\alpha_j^t)^2
\]
and
\[
M_{\mathrm{eff}}^t
=
\frac{(c^t)^2}{\sum_{j\ne i}(\alpha_j^t)^4}.
\]
Both quantities are invariant under the sign change
$\alpha^t\mapsto -\alpha^t$. Hence
\[
\mathbf{1}\{\mathcal{A}^t_i(\mathcal{Q}^t)\}
=
\mathbf{1}\{\mathcal{A}^t_i(T(\mathcal{Q}^t))\}.
\]
Since $T$ preserves the proposal distribution and preserves the acceptance
event, the accepted query design is also invariant under $T$:
\[
\mathcal{Q}^t\mid \mathcal{A}^t_i
\;\stackrel{d}{=}\;
T(\mathcal{Q}^t)\mid \mathcal{A}^t_i.
\]
Consequently, for every $j\ne i$,
\[
\mathbb{E}[\alpha_j^t\mid \mathcal{A}^t_i]
=
\mathbb{E}[\alpha_j^t(T(\mathcal{Q}^t))\mid \mathcal{A}^t_i]
=
-\mathbb{E}[\alpha_j^t\mid \mathcal{A}^t_i],
\]
which implies
\[
\mathbb{E}[\alpha_j^t\mid \mathcal{A}^t_i]=0.
\]

Finally, the query design $\mathcal{Q}^t$ is sampled independently of
the client updates $\Delta_1^t,\ldots,\Delta_K^t$ and the
current global model, and the acceptance event $\mathcal{A}^t_i$ is a
function only of the query design. Therefore the same identity holds
after conditioning on the realized updates:
\[
\mathbb{E}
[
\alpha_j^t
\mid
\Delta_1^t,\ldots,\Delta_K^t,\mathcal{A}^t_i
]
=0,
\qquad j\ne i .
\]
Thus
\[
\begin{aligned}
\mathbb{E}
\left[
\widehat{\Delta}_i^{\,t}
\mid
\Delta_1^t,\ldots,\Delta_K^t,\mathcal{A}^t_i
\right]
&=
\Delta_i^t
+
\sum_{j\ne i}
\mathbb{E}
[
\alpha_j^t
\mid
\Delta_1^t,\ldots,\Delta_K^t,\mathcal{A}^t_i
]
\Delta_j^t \\
&=
\Delta_i^t .
\end{aligned}
\]
This proves the unbiasedness of the accepted estimator.
\end{proof}

\subsection{Proof of Theorem~\ref{thm:variance} 
(Conditional variance under rejection sampling)}
\label{app:variance}

\begin{proof}
Let $\zeta(\mathcal{Q}^t) := \sum_{j \neq i} 
\alpha_j^t\,\Delta_j^t = 
\widehat{\Delta}_i^{\,t} - \Delta_i^t$. By 
Theorem~\ref{thm:unbiased}, 
$\mathbb{E}[\zeta \mid \mathcal{A}^t_i, 
\Delta_1^t, \ldots, \Delta_K^t] = 0$, so the 
accepted covariance is
\[
\operatorname{Cov}(\widehat{\Delta}_i^{\,t} \mid 
\mathcal{A}^t_i, \Delta_1^t, \ldots, \Delta_K^t)
= \mathbb{E}[\zeta\zeta^\top \mid \mathcal{A}^t_i]
= \frac{1}{p_a}\,
\mathbb{E}[\mathbf{1}_{\mathcal{A}}\,
\zeta\zeta^\top]
\preceq \frac{1}{p_a}\,
\mathbb{E}[\zeta\zeta^\top],
\]
where $p^t_{a,i} = \Pr(\mathcal{A}^t_i)$. Each 
non-target part $U_m^t \setminus \{i\}$ is a 
simple random sample of size $N$ drawn without 
replacement from the $K-1$ non-target clients. By 
the finite-population correction 
formula~\citep{cochran1977sampling}, each 
include-target sum 
$A_m := \sum_{j \in U_m^t \setminus \{i\}} 
\Delta_j^t$ has covariance 
$\operatorname{Cov}(A_m) = N(K-1-N)/(K-2)\,
\Sigma_{-i}^t$, and likewise for each 
exclude-target sum. The $2M$ sums are mutually 
independent, so
\[
\mathbb{E}[\zeta\zeta^\top]
= \frac{2}{M} \cdot \frac{N(K-1-N)}{K-2}\,
\Sigma_{-i}^t.
\]
Combining gives
\[
\operatorname{Cov}(\widehat{\Delta}_i^{\,t} \mid 
\mathcal{A}^t_i, \Delta_1^t, \ldots, \Delta_K^t)
\preceq \frac{1}{p_a} \cdot \frac{2}{M} 
\cdot \frac{N(K-1-N)}{K-2}\,\Sigma_{-i}^t.
\qedhere
\]
\end{proof}

\subsection{Proof of Theorem~\ref{thm:stouffer} (Stouffer concentration)}
\label{app:stouffer}

\begin{proof}
Under Assumption~\ref{asm:separation},
write $z_i^{(t)} = \mu_i^{(t)} + \eta_i^{(t)}$
where $\eta_i^{(t)}$ is conditionally $\nu^2$-sub-Gaussian given $\mathcal{F}_{t-1}$.
Define $S_T := \sum_{t=1}^T \eta_i^{(t)}$.
By the conditional sub-Gaussian tower property,
$\mathbb{E}[\exp(\lambda S_T)] \leq \exp(T\lambda^2\nu^2/2)$.
Hence $W_i := T^{-1/2} S_T$ is $\nu^2$-sub-Gaussian.

\paragraph{Clean client.}
Since $|\mu_i^{(t)}| \leq \epsilon$,
$Z_i \leq \sqrt{T}\,\epsilon + W_i$,
so $\Pr(Z_i \geq \gamma) \leq \exp(-(\gamma - \sqrt{T}\,\epsilon)^2 / (2\nu^2))$
for $\gamma > \sqrt{T}\,\epsilon$.

\paragraph{Watermarked client.}
Since $\mu_i^{(t)} \geq m$,
$Z_i \geq \sqrt{T}\,m + W_i$,
so $\Pr(Z_i \leq \gamma) \leq \exp(-(\sqrt{T}\,m - \gamma)^2 / (2\nu^2))$
for $\sqrt{T}\,m > \gamma$.
\end{proof}

\section{Proofs for Privacy Analysis}
\label{app:privacy}

This section adapts the proof strategy of \citet{elkordy2023howmuch} to the
client-level update estimates produced by FedAttr. We use the same data IID assumption and 
distributional assumptions as the independent-under-whitening case in \citet{elkordy2023howmuch}. The difference of proof structure is that the equal-weight secure aggregation is replaced by a query-dependent estimation. Accordingly, the number of masking users in \citet{elkordy2023howmuch} is replaced by the effective
masking size $M_{\mathrm{eff},i}^t$, where
\[
  M_{\mathrm{eff},i}^t
  =
  \frac{\left(c_i^t\right)^2}{s_i^t},
  \qquad
  c_i^t:=\sum_{j\ne i}(\alpha_j^t)^2,
  \qquad
  s_i^t:=\sum_{j\ne i}(\alpha_j^t)^4.
\]

The proof follows the same high-level route as the independent-under-whitening
case of \citet{elkordy2023howmuch} in 4 stages:
\begin{enumerate}
  \item Decompose the mutual information into a difference of two entropies.
  \item Upper-bound the entropy of signal plus noise by Gaussian maximum
  entropy.
  \item Lower-bound the entropy of the masking noise using the
  Bobkov--Chistyakov--G{\"o}tze entropic Berry--Esseen bound
  \cite{bobkov2014berry}.
  \item Subtract the two bounds.
\end{enumerate}

\subsection{FedAttr Notation}
\label{sec:notation}
For convenience, we first list the notations in \textbf{FedAttr}.

Given a communication round \(t\) and a target client \(i\).  FedAttr protocol queries are
\[
  \mathcal Q_i^t=(U_1^t,\ldots,U_M^t;V_1^t,\ldots,V_M^t),
\]
where each include subset \(U_m^t\) contains \(i\), and each exclude subset
\(V_m^t\) does not contain \(i\).  Given a secure aggregate
\(S^t(W)=\sum_{j\in W}\Delta_j^t\), the released estimate for target $i$ is
\[
  \widehat\Delta_i^t
  :=
  \frac1M\sum_{m=1}^M S^t(U_m^t)
  -
  \frac1M\sum_{m=1}^M S^t(V_m^t).
\]
Expanding this linear combination gives
\begin{equation}
  \widehat\Delta_i^t
  =
  \Delta_i^t+
  \sum_{j\ne i}\alpha_j^t(\mathcal Q_i^t)\Delta_j^t,
\end{equation}
where
\[
  \alpha_j^t(\mathcal Q_i^t)
  :=
  \frac1M\sum_{m=1}^M {\bf 1}\{j\in U_m^t\}
  -
  \frac1M\sum_{m=1}^M {\bf 1}\{j\in V_m^t\}.
\]
After conditioning on \(\mathcal Q_i^t\), the coefficients \(\alpha_j^t\) are
deterministic. For convenience, given \(\mathcal Q_i^t\), we notate \(\alpha_j^t:=\alpha_j^t(\mathcal Q_i^t)\). Therefore, 

\begin{equation}
\label{eq:fedattr-surrogate-expansion}
  \widehat\Delta_i^t
  =
  \Delta_i^t+
  \sum_{j\ne i}\alpha_j^t\Delta_j^t,
\end{equation}
where
\[
  \alpha_j^t
  :=
  \frac1M\sum_{m=1}^M {\bf 1}\{j\in U_m^t\}
  -
  \frac1M\sum_{m=1}^M {\bf 1}\{j\in V_m^t\}.
\]

\begin{definition}
\label{def:masking-quantities}
Given queries \(Q_i^t\), define
\[
  c_i^t
  :=
  \sum_{j\ne i}\bigl(\alpha_j^t\bigr)^2,
\]
\[
  s_i^t
  :=
  \sum_{j\ne i}\bigl(\alpha_j^t\bigr)^4.
\]
When \(c_i^t>0\), define
\[
  M_{\mathrm{eff},i}^t
  :=
  \frac{\bigl(c_i^t\bigr)^2}{s_i^t}.
\]
Equivalently, if
\[
  \beta_j^t:=\frac{\alpha_j^t}{\sqrt{c_i^t}},
\]
then
\[
  \sum_{j\ne i}(\beta_j^t)^2=1,
  \qquad
  M_{\mathrm{eff},i}^t=\frac{1}{\sum_{j\ne i}(\beta_j^t)^4}.
\]
\end{definition}

\subsection{Stage 1: Two Entropies Decomposition}
\label{sec:step1}

In stage 1, we decompose the mutual information into the difference of two entropies.

\begin{lemma}[FedAttr estimate decomposition]
\label{lem:signal-noise}
Under Assumption~\ref{asm:iid-fluctuations}, conditioned on
\((Q_i^t,\F_{t-1})\), the FedAttr estimate can be written as
\[
  \widehat\Delta_i^t
  =
  \left(1+\sum_{j\ne i}\alpha_j^t\right)\mu^t
  +
  \xi_i^t+
  \eta_i^t,
\]
where
\[
  \eta_i^t:=\sum_{j\ne i}\alpha_j^t\xi_j^t.
\]
Moreover, \(\xi_i^t\) is conditionally independent of \(\eta_i^t\) given
\((Q_i^t,\F_{t-1})\).  Consequently,
\begin{equation}
\label{eq:mi-entropy-diff}
  I(\Delta_i^t;\widehat\Delta_i^t\mid Q_i^t,\F_{t-1})
  =
  h(\xi_i^t+\eta_i^t\mid Q_i^t,\F_{t-1})
  -
  h(\eta_i^t\mid Q_i^t,\F_{t-1}).
\end{equation}
\end{lemma}

\begin{proof}
Substitute \(\Delta_j^t=\mu^t+\xi_j^t\) into
\eqref{eq:fedattr-surrogate-expansion}:
\[
\begin{aligned}
  \widehat\Delta_i^t
  &=
  \mu^t+\xi_i^t+
  \sum_{j\ne i}\alpha_j^t(\mu^t+\xi_j^t) \\
  &=
  \left(1+\sum_{j\ne i}\alpha_j^t\right)\mu^t
  +
  \xi_i^t+
  \sum_{j\ne i}\alpha_j^t\xi_j^t.
\end{aligned}
\]
This proves the stated decomposition.

Conditioned on $(Q_i^t,\mathcal F_{t-1})$, the coefficients
$\alpha_j^t$ are deterministic. The vector $\xi_i^t$ is client $i$'s
centered update, while
\[
  \eta_i^t
  =
  \sum_{j\ne i}\alpha_j^t \xi_j^t
\]
is a deterministic function of the centered updates of the non-target
clients $\{\xi_j^t:j\ne i\}$. By Assumption~\ref{asm:iid-fluctuations},
the centered updates $\{\xi_j^t\}_{j=1}^K$ are conditionally independent
given $\mathcal F_{t-1}$. Therefore, $\xi_i^t$ is conditionally
independent of the collection $\{\xi_j^t:j\ne i\}$, and hence is
conditionally independent of any deterministic function of the
collection $\{\xi_j^t:j\ne i\}$, including $\eta_i^t$.

The deterministic shift
\[
  \left(1+\sum_{j\ne i}\alpha_j^t\right)\mu^t
\]
does not affect mutual information.  Since \(\Delta_i^t=\mu^t+\xi_i^t\),
we have
\[
  I(\Delta_i^t;\widehat\Delta_i^t\mid Q_i^t,\F_{t-1})
  =
  I(\xi_i^t;\xi_i^t+\eta_i^t\mid Q_i^t,\F_{t-1}).
\]
For independent \(X\) and \(Z\),
\[
  I(X;X+Z)=h(X+Z)-h(X+Z\mid X)=h(X+Z)-h(Z).
\]
Applying this identity with \(X=\xi_i^t\) and \(Z=\eta_i^t\) gives
\eqref{eq:mi-entropy-diff}.
\end{proof}

\subsection{Stage 2: Upper Bound the $h(\xi_i^t+\eta_i^t\mid Q_i^t,\F_{t-1})$ by Gaussian Maximum Entropy.}
\label{sec:step2}

In this stage, we use the Gaussian maximum entropy to upper-bound the $h(\xi_i^t+\eta_i^t\mid Q_i^t,\F_{t-1})$ (first term).

\begin{lemma}[Gaussian maximum-entropy upper bound]
\label{lem:upper-signal-noise}
Under Assumption~\ref{asm:iid-fluctuations}, given queries $Q_i^t$ with
\(c_i^t:=
  \sum_{j\ne i}\bigl(\alpha_j^t\bigr)^2>0\),
\[
  h(\xi_i^t+\eta_i^t\mid Q_i^t,\F_{t-1})
  \le
  \frac12\log\det\bigl(2\pi e(1+c_i^t)K_G^t\bigr).
\]
\end{lemma}

\begin{proof}
Because \(\xi_i^t\) is conditionally independent of \(\eta_i^t\), covariance
adds:
\[
  \Cov(\xi_i^t+\eta_i^t\mid Q_i^t,\F_{t-1})
  =
  \Cov(\xi_i^t\mid\F_{t-1})
  +
  \Cov(\eta_i^t\mid Q_i^t,\F_{t-1}).
\]
The first term equals \(K_G^t\).  For the second term,
\[
\begin{aligned}
  \Cov(\eta_i^t\mid Q_i^t,\F_{t-1})
  &=
  \Cov\left(\sum_{j\ne i}\alpha_j^t\xi_j^t\mid Q_i^t,\F_{t-1}\right) \\
  &=
  \sum_{j\ne i}(\alpha_j^t)^2K_G^t
  =
  c_i^tK_G^t,
\end{aligned}
\]
where the cross-covariances vanish because of conditional independence.  Therefore,
\[
  \Cov(\xi_i^t+\eta_i^t\mid Q_i^t,\F_{t-1})=(1+c_i^t)K_G^t.
\]
Among all distributions with the same covariance matrix \(\Sigma\), the
Gaussian has the largest differential entropy, equal to
\(\frac12\log\det(2\pi e\Sigma)\).  Taking
\(\Sigma=(1+c_i^t)K_G^t\) proves the claim.
\end{proof}

\subsection{Step 3: Lower Bound the $h(\eta_i^t\mid Q_i^t,\F_{t-1})$ by Entropic Berry–Esseen Bound}
\label{sec:step3}

In this stage, we lower bound the $h(\eta_i^t\mid Q_i^t,\F_{t-1})$ by entropic Berry–Esseen bound. We first introduce a lemma to show scalar weighted entropic Berry--Esseen bound.

\begin{lemma}[Scalar weighted entropic Berry--Esseen bound]
\label{lem:scalar-bcg}
Let \(X_1,\ldots,X_m\) be independent centered scalar random variables.
Omit any zero-variance summands, and assume the remaining summands have
positive variances, finite fourth moments, densities, and finite
differential entropies. Let
\[
  V_m:=\sum_{r=1}^m \mathrm{Var}(X_r),
\]
and assume \(V_m=1\). For each \(r\), let \(Z_r\) be a Gaussian random
variable with the same mean and variance as \(X_r\). Assume the
Bobkov--Chistyakov--G\"otze entropic Berry--Esseen regularity holds
uniformly; in particular, assume there exists \(D_0<\infty\) such that
\[
  \KL(X_r\|Z_r)=h(Z_r)-h(X_r)\le D_0
  \qquad\text{for every }r.
\]
Let \(G\sim\mathcal N(0,1)\). Then there is a constant
\(C_{\rm BCG}\), depending only on the corresponding BCG regularity
constants, such that
\[
  \KL\left(\sum_{r=1}^m X_r\,\middle\|\,G\right)
  \le
  C_{\rm BCG}\sum_{r=1}^m \E |X_r|^4.
\]
Equivalently,
\[
  h\left(\sum_{r=1}^m X_r\right)
  \ge
  \frac12\log(2\pi e)
  -
  C_{\rm BCG}\sum_{r=1}^m \E |X_r|^4.
\]
\end{lemma}

This lemma does not introduce a FedAttr-specific modeling assumption.
The condition
\[
  \KL(X_r\|Z_r)=h(Z_r)-h(X_r)\le D_0
\]
is part of the regularity needed by the Bobkov--Chistyakov--G\"otze
entropic Berry--Esseen theorem. In the independent-under-whitening case,
\citet{elkordy2023howmuch} use this entropic Berry--Esseen tool to
lower-bound the entropy of an equal-weight normalized aggregate; the
corresponding regularity is absorbed into their constant. We state it
explicitly because FedAttr applies the same tool to query-dependent
weighted summands, which are independent but not necessarily identically
distributed.

\begin{proof}
This is the one-dimensional entropic Berry--Esseen theorem of
Bobkov--Chistyakov--G{\"o}tze for independent, not necessarily identically
distributed, summands.  In their notation, for
\[
  S_m:=\frac{\sum_{r=1}^m X_r}{\sqrt{V_m}},
\]
the entropic distance from \(S_m\) to the standard Gaussian is bounded by a
constant depending on the uniform entropic-distance parameter, times the
Lyapunov fourth-moment ratio
\[
  \frac{\sum_{r=1}^m \E |X_r|^4}{V_m^2}.
\]
Since \(V_m=1\), this gives
\[
  \KL\left(\sum_{r=1}^m X_r\,\middle\|\,G\right)
  \le
  C_{\rm BCG}\sum_{r=1}^m \E |X_r|^4.
\]

It remains only to translate the relative-entropy statement into an entropy
lower bound.  Let
\[
  S:=\sum_{r=1}^m X_r.
\]
Then \(\E S=0\) and \(\mathrm{Var}(S)=1\).  The density of
\(G\sim\mathcal N(0,1)\) is
\[
  \phi(x)=\frac1{\sqrt{2\pi}}\exp(-x^2/2).
\]
Thus
\[
\begin{aligned}
  \KL(S\|G)
  &=
  \int p_S(x)\log\frac{p_S(x)}{\phi(x)}\,dx \\
  &=
  -h(S)-\E[\log\phi(S)].
\end{aligned}
\]
Since
\[
  \log\phi(x)=-\frac12\log(2\pi)-\frac{x^2}{2},
\]
and \(\E S^2=1\),
\[
  \E[\log\phi(S)]
  =
  -\frac12\log(2\pi)-\frac12.
\]
Therefore
\[
  \KL(S\|G)=\frac12\log(2\pi e)-h(S).
\]
Rearranging the BCG bound yields the claimed entropy lower bound.
\end{proof}

\begin{lemma}[Coordinate tensorization]
\label{lem:tensorization}
Condition on \(\mathcal G:=(Q_i^t,\F_{t-1})\), and suppose
Assumption~\ref{asm:bcg-regularity} holds.  Let
\[
  S_\beta^t:=\sum_{j\ne i}\beta_j^t Z_j^t,
  \qquad
  \sum_{j\ne i}(\beta_j^t)^2=1.
\]
Then the coordinates of \(S_\beta^t\) are conditionally independent given
\(\mathcal G\), and
\[
  h(S_\beta^t\mid\mathcal G)
  =
  \sum_{\ell=1}^{d^*}h(S_{\beta,\ell}^t\mid\mathcal G),
\]
where
\[
  S_{\beta,\ell}^t:=\sum_{j\ne i}\beta_j^t Z_{j,\ell}^t.
\]
\end{lemma}

\begin{proof}
After conditioning on \(\mathcal G\), the coefficients \(\beta_j^t\) are
deterministic.  By Assumption~\ref{asm:bcg-regularity}, each vector
\(Z_j^t=(Z_{j,1}^t,\ldots,Z_{j,d^*}^t)\) has independent coordinates, and by
Assumption~\ref{asm:iid-fluctuations}, the vectors \(Z_j^t\) are independent
across \(j\).  Hence the full scalar collection
\[
  \{Z_{j,\ell}^t:j\ne i,\ell\in[d^*]\}
\]
has a joint density that factorizes as
\[
  \prod_{j\ne i}\prod_{\ell=1}^{d^*}f_\ell(z_{j,\ell}),
\]
where the same coordinate density \(f_\ell\) is used across clients because
the fluctuations are conditionally i.i.d.

For a fixed coordinate \(\ell\), the random variable
\[
  S_{\beta,\ell}^t=\sum_{j\ne i}\beta_j^tZ_{j,\ell}^t
\]
depends only on the collection \(\{Z_{j,\ell}^t:j\ne i\}\).  The collections
corresponding to different coordinates are independent because the joint
density factorizes over \(\ell\).  Therefore
\(S_{\beta,1}^t,\ldots,S_{\beta,d^*}^t\) are conditionally independent.

If \(p_\ell\) is the density of \(S_{\beta,\ell}^t\), the joint density of
\(S_\beta^t\) is
\[
  p(s_1,\ldots,s_{d^*})=\prod_{\ell=1}^{d^*}p_\ell(s_\ell).
\]
Thus
\[
\begin{aligned}
  h(S_\beta^t\mid\mathcal G)
  &=
  -\int \prod_{\ell=1}^{d^*}p_\ell(s_\ell)
  \log\left(\prod_{\ell=1}^{d^*}p_\ell(s_\ell)\right)ds \\
  &=
  -\sum_{\ell=1}^{d^*}\int \prod_{r=1}^{d^*}p_r(s_r)
  \log p_\ell(s_\ell)ds \\
  &=
  -\sum_{\ell=1}^{d^*}\int p_\ell(s_\ell)\log p_\ell(s_\ell)ds_\ell \\
  &=
  \sum_{\ell=1}^{d^*}h(S_{\beta,\ell}^t\mid\mathcal G).
\end{aligned}
\]
\end{proof}

\begin{lemma}[Weighted entropy lower bound for the normalized mask]
\label{lem:weighted-vector-entropy}
Under Assumptions~\ref{asm:iid-fluctuations} and
\ref{asm:bcg-regularity}, for any fixed query design with \(c_i^t>0\),
\[
  h(S_\beta^t\mid Q_i^t,\F_{t-1})
  \ge
  \frac{d^*}{2}\log(2\pi e)
  -
  \frac{C_\xi d^*}{M_{\mathrm{eff},i}^t},
\]
where \(C_\xi\) depends only on the one-dimensional regularity constants in
Assumption~\ref{asm:bcg-regularity}.
\end{lemma}

\begin{proof}
By Lemma~\ref{lem:tensorization},
\[
  h(S_\beta^t\mid Q_i^t,\F_{t-1})
  =
  \sum_{\ell=1}^{d^*}h(S_{\beta,\ell}^t\mid Q_i^t,\F_{t-1}).
\]
Fix a coordinate \(\ell\).  Define the scalar summands
\[
  X_j:=\beta_j^tZ_{j,\ell}^t,
  \qquad j\ne i.
\]
They are independent, centered, and their total variance is
\[
  \sum_{j\ne i}\mathrm{Var}(X_j)
  =
  \sum_{j\ne i}(\beta_j^t)^2\mathrm{Var}(Z_{j,\ell}^t)
  =
  \sum_{j\ne i}(\beta_j^t)^2
  =1.
\]
Zero weights can be removed from the sum, so the scalar BCG bound applies to
the nonzero summands.  Moreover,
\[
\begin{aligned}
  \sum_{j\ne i}\E |X_j|^4
  &=
  \sum_{j\ne i}(\beta_j^t)^4\E |Z_{j,\ell}^t|^4 \\
  &\le
  M_{4,\ell}\sum_{j\ne i}(\beta_j^t)^4
  =
  \frac{M_{4,\ell}}{M_{\mathrm{eff},i}^t}.
\end{aligned}
\]
Scaling by \(\beta_j^t\) does not change the entropic distance to the matching
Gaussian for nonzero \(\beta_j^t\), because both the variable and its matching
Gaussian are transformed by the same invertible scalar map.  Therefore
Lemma~\ref{lem:scalar-bcg} yields
\[
  h(S_{\beta,\ell}^t\mid Q_i^t,\F_{t-1})
  \ge
  \frac12\log(2\pi e)
  -
  \frac{C_\ell}{M_{\mathrm{eff},i}^t},
\]
for a constant \(C_\ell\) depending on \(D_{0,\ell}\) and \(M_{4,\ell}\).
Let \(C_\xi:=\max_\ell C_\ell\).  Summing over \(\ell=1,\ldots,d^*\),
\[
\begin{aligned}
  h(S_\beta^t\mid Q_i^t,\F_{t-1})
  &\ge
  \sum_{\ell=1}^{d^*}
  \left(\frac12\log(2\pi e)-\frac{C_\ell}{M_{\mathrm{eff},i}^t}\right) \\
  &\ge
  \frac{d^*}{2}\log(2\pi e)-\frac{C_\xi d^*}{M_{\mathrm{eff},i}^t}.
\end{aligned}
\]
\end{proof}

\begin{lemma}[Entropy lower bound for the FedAttr masking noise]
\label{lem:mask-entropy}
Under Assumptions~\ref{asm:iid-fluctuations} and
\ref{asm:bcg-regularity}, for any fixed query design with \(c_i^t>0\),
\[
  h(\eta_i^t\mid Q_i^t,\F_{t-1})
  \ge
  \frac{d^*}{2}\log(2\pi e\,c_i^t)
  +
  \frac12\log\det K_G^t
  -
  \frac{C_\xi d^*}{M_{\mathrm{eff},i}^t}.
\]
\end{lemma}

\begin{proof}
Since \(\xi_j^t=(K_G^t)^{1/2}Z_j^t\),
\[
\begin{aligned}
  \eta_i^t
  &=
  \sum_{j\ne i}\alpha_j^t\xi_j^t \\
  &=
  (K_G^t)^{1/2}\sum_{j\ne i}\alpha_j^tZ_j^t \\
  &=
  \sqrt{c_i^t}(K_G^t)^{1/2}
  \sum_{j\ne i}\beta_j^tZ_j^t \\
  &=
  \sqrt{c_i^t}(K_G^t)^{1/2}S_\beta^t.
\end{aligned}
\]
For an invertible matrix \(A\), differential entropy satisfies
\[
  h(AX)=h(X)+\log|\det A|.
\]
Applying this identity to
\[
  A=\sqrt{c_i^t}(K_G^t)^{1/2},
\]
we obtain
\[
  h(\eta_i^t\mid Q_i^t,\F_{t-1})
  =
  h(S_\beta^t\mid Q_i^t,\F_{t-1})
  +
  \frac{d^*}{2}\log c_i^t
  +
  \frac12\log\det K_G^t.
\]
Substituting Lemma~\ref{lem:weighted-vector-entropy} gives the result.
\end{proof}

\subsection{Step 4: Subtract the Two Bounds.}
\label{sec:step4}

\begin{theorem}[Fixed-query release-level MI leakage]
\label{thm:fixed-query}
Under Assumptions~\ref{asm:bcg-regularity} and
\ref{asm:iid-fluctuations}, for any fixed query design \(Q_i^t\), independent
of the updates, with \(c_i^t(Q_i^t)>0\),
\[
  I(\Delta_i^t;\widehat\Delta_i^t\mid Q_i^t,\F_{t-1})
  \le
  \frac{d^*}{2}\log\left(1+\frac1{c_i^t}\right)
  +
  \frac{C_\xi d^*}{M_{\mathrm{eff},i}^t}.
\]
\end{theorem}

\begin{proof}
Start from Lemma~\ref{lem:signal-noise}:
\[
  I(\Delta_i^t;\widehat\Delta_i^t\mid Q_i^t,\F_{t-1})
  =
  h(\xi_i^t+\eta_i^t\mid Q_i^t,\F_{t-1})
  -
  h(\eta_i^t\mid Q_i^t,\F_{t-1}).
\]
By Lemma~\ref{lem:upper-signal-noise},
\[
  h(\xi_i^t+\eta_i^t\mid Q_i^t,\F_{t-1})
  \le
  \frac12\log\det\bigl(2\pi e(1+c_i^t)K_G^t\bigr).
\]
Expanding the determinant on the \(d^*\)-dimensional effective subspace,
\[
  \frac12\log\det\bigl(2\pi e(1+c_i^t)K_G^t\bigr)
  =
  \frac{d^*}{2}\log(2\pi e(1+c_i^t))
  +
  \frac12\log\det K_G^t.
\]
By Lemma~\ref{lem:mask-entropy},
\[
  h(\eta_i^t\mid Q_i^t,\F_{t-1})
  \ge
  \frac{d^*}{2}\log(2\pi e\,c_i^t)
  +
  \frac12\log\det K_G^t
  -
  \frac{C_\xi d^*}{M_{\mathrm{eff},i}^t}.
\]
Subtracting the lower bound on \(h(\eta_i^t)\) from the upper bound on
\(h(\xi_i^t+\eta_i^t)\), the terms
\(\frac12\log\det K_G^t\) and \(\frac{d^*}{2}\log(2\pi e)\) cancel.  Hence
\[
\begin{aligned}
  I(\Delta_i^t;\widehat\Delta_i^t\mid Q_i^t,\F_{t-1})
  &\le
  \frac{d^*}{2}\log\left(\frac{1+c_i^t}{c_i^t}\right)
  +
  \frac{C_\xi d^*}{M_{\mathrm{eff},i}^t} \\
  &=
  \frac{d^*}{2}\log\left(1+\frac1{c_i^t}\right)
  +
  \frac{C_\xi d^*}{M_{\mathrm{eff},i}^t}.
\end{aligned}
\]
\end{proof}

\subsection{Proof of Theorem~\ref{thm:diffuse-query} 
(Release-level MI leakage)}
\begin{proof}
Theorem~\ref{thm:fixed-query} gives
\[
  I
  \le
  \frac{d^*}{2}\log\left(1+\frac1{c_i^t}\right)
  +
  \frac{C_\xi d^*}{M_{\mathrm{eff},i}^t}.
\]
On the acceptance event,
\[
  c_i^t\ge aN,
  \qquad
  M_{\mathrm{eff},i}^t\ge aN.
\]
Since \(x\mapsto \log(1+1/x)\) is decreasing for \(x>0\),
\[
  \log\left(1+\frac1{c_i^t}\right)
  \le
  \log\left(1+\frac1{aN}\right),
\]
and
\[
  \frac1{M_{\mathrm{eff},i}^t}\le \frac1{aN}.
\]
Substitution proves the displayed bound.  The order statement follows from
\(\log(1+x)\le x\) for \(x\ge0\).
\end{proof}

\section{Sufficient-condition Analysis for Differential Scoring}
\label{app:differential-scoring}

This section connects the estimator guarantees in
Theorems~\ref{thm:unbiased}--\ref{thm:variance} to the score-separation
condition used by the Stouffer analysis in Theorem~\ref{thm:stouffer}.  The
main text uses the differential score
\[
    z_i^{(t)}
    =
    \mathrm{SCORE}(w^{t-1}+\widehat\Delta_i^t;P_t)
    -
    \mathrm{SCORE}(w^{t-1};P_t)
\]
as the per-round evidence for client-level attribution.  However,
\(z_i^{(t)}\) is computed from the SA-based estimate
\(\widehat\Delta_i^t\), rather than from the true client update
\(\Delta_i^t\).  The purpose of this section is to show that, under local
regularity of the score function, this observed differential score is close to
the oracle single-client differential score
\[
    \psi_i^{(t)}
    :=
    F_t(w^{t-1}+\Delta_i^t)-F_t(w^{t-1}),
\]
with an error controlled by the variance of the paired-subset estimator.

The argument has three steps.  First, differential scoring exactly cancels any
additive score baseline shared by all clients in round \(t\), explaining why
subtracting \(F_t(w^{t-1})\) removes the watermark signal already accumulated
in the global model.  Second, writing
\(\zeta_i^t:=\widehat\Delta_i^t-\Delta_i^t\), the only difference between the
FedAttr score and the oracle score is
\[
    z_i^{(t)}-\psi_i^{(t)}
    =
    F_t(w^{t-1}+\Delta_i^t+\zeta_i^t)
    -
    F_t(w^{t-1}+\Delta_i^t).
\]
Third, the Lipschitz or smoothness regularity of \(F_t\), together with the
accepted-law unbiasedness and variance bound of
\(\widehat\Delta_i^t\), bounds this score error.  Consequently, if the oracle
differential scores separate watermarked and benign clients, then the observed
FedAttr differential scores inherit the same separation up to a
variance-controlled error term \(R\).

This section therefore gives a sufficient-condition analysis for the
mean-separation part of Assumption~\ref{asm:separation}.  It does not prove
that the watermark detector separates clients unconditionally, nor does it
prove the sub-Gaussian residual condition; the latter remains the
score-level condition used in Theorem~\ref{thm:stouffer} and is empirically
validated in Figure~\ref{fig:mechanism}.

Throughout this section,
the prompt set \(P_t\) is treated as fixed in round \(t\).  We write
\[
    F_t(w):=\mathrm{SCORE}(w;P_t).
\]
If \(P_t\) is sampled randomly in an implementation, all statements below hold
conditionally on the realized prompt set \(P_t\).

Recall that FedAttr computes
\[
    z_i^{(t)}
    =
    F_t(w^{t-1}+\widehat\Delta_i^t)-F_t(w^{t-1}).
\]
The goal of this step is to remove the watermark baseline already present in
the current global model \(w^{t-1}\), and to isolate the incremental
contribution of client \(i\)'s current-round update.

\paragraph{Accepted-query convention.}
Let \(A_i^t\) denote the accepted-query event for target client \(i\):
\[
    A_i^t
    :=
    \left\{
        c_i^t\ge aN,\quad M_{\mathrm{eff},i}^t\ge aN
    \right\}.
\]
Assume \(p_{a,i}^t:=\Pr(A_i^t)>0\).  Since FedAttr resamples query designs
until \(A_i^t\) holds, every released estimator
\(\widehat\Delta_i^t\) and every released score \(z_i^{(t)}\) is generated
under the accepted-query distribution.  Equivalently, expectations involving
\(\widehat\Delta_i^t\) or \(z_i^{(t)}\) are conditional on \(A_i^t\).  We use
the shorthand
\[
    \mathbb E_a[\cdot\mid\mathcal H]
    :=
    \mathbb E[\cdot\mid\mathcal H,A_i^t],
\]
for any conditioning sigma-field \(\mathcal H\) not containing the current
query draw.

Define the \emph{oracle differential score} that would be obtained if the
true client update \(\Delta_i^t\) were available:
\begin{equation}
\label{eq:oracle-diff-score-simple}
    \psi_i^{(t)}
    :=
    F_t(w^{t-1}+\Delta_i^t)-F_t(w^{t-1}).
\end{equation}
Let
\begin{equation}
\label{eq:update-error-simple}
    \zeta_i^t
    :=
    \widehat\Delta_i^t-\Delta_i^t
\end{equation}
be the update-estimation error.  Then
\begin{equation}
\label{eq:score-error-decomp-simple}
    z_i^{(t)}-\psi_i^{(t)}
    =
    F_t(w^{t-1}+\Delta_i^t+\zeta_i^t)
    -
    F_t(w^{t-1}+\Delta_i^t).
\end{equation}
Thus the gap between FedAttr's observed differential score and the oracle
single-client differential score is caused only by the update-estimation error
\(\zeta_i^t\).

\begin{lemma}[Exact cancellation of additive baselines]
\label{lem:baseline-cancellation-simple}
Let \(b_t\) be any scalar depending only on the round \(t\), the prompt set
\(P_t\), and the past history.  Define a shifted score
\[
    \widetilde F_t(w):=F_t(w)+b_t.
\]
Then direct scoring is shifted by \(b_t\), while differential scoring is
unchanged:
\[
\begin{aligned}
    \widetilde F_t(w^{t-1}+\widehat\Delta_i^t)
    -
    \widetilde F_t(w^{t-1})
    &=
    F_t(w^{t-1}+\widehat\Delta_i^t)
    -
    F_t(w^{t-1}).
\end{aligned}
\]
\end{lemma}

\begin{proof}
By direct subtraction,
\[
\begin{aligned}
    \widetilde F_t(w^{t-1}+\widehat\Delta_i^t)
    -
    \widetilde F_t(w^{t-1})
    &=
    \bigl(F_t(w^{t-1}+\widehat\Delta_i^t)+b_t\bigr)
    -
    \bigl(F_t(w^{t-1})+b_t\bigr)\\
    &=
    F_t(w^{t-1}+\widehat\Delta_i^t)
    -
    F_t(w^{t-1}).
\end{aligned}
\]
\end{proof}

\begin{assumption}[Local regularity of the score]
\label{asm:score-regularity-simple}
For each round \(t\), the score function \(F_t\) is locally regular on the
region visited by FedAttr.  Specifically, there exists \(L_t<\infty\) such
that, for every target client \(i\),
\[
    |F_t(x)-F_t(y)|\le L_t\|x-y\|
\]
for all points \(x,y\) on the line segment between
\(w^{t-1}+\Delta_i^t\) and \(w^{t-1}+\widehat\Delta_i^t\).

When the second-order bound is invoked, we further assume that \(F_t\) is
differentiable and has \(H_t\)-Lipschitz gradient on the same local region:
\[
    \|\nabla F_t(x)-\nabla F_t(y)\|
    \le
    H_t\|x-y\|.
\]
\end{assumption}

\begin{theorem}[Approximation of oracle differential scores]
\label{thm:differential-score-approx-simple}
Fix a round \(t\) and target client \(i\).  Condition on the past
\(\mathcal F_{t-1}\) and on the realized client updates
\(\Delta_1^t,\ldots,\Delta_K^t\).  Define
\begin{equation}
\label{eq:score-variance-proxy-simple}
    B_i^t
    :=
    \frac{1}{p_{a,i}^t}
    \cdot
    \frac{2}{M}
    \cdot
    \frac{N(K-1-N)}{K-2}
    \operatorname{tr}(\Sigma_{-i}^t).
\end{equation}
Under Assumption~\ref{asm:score-regularity-simple},
\begin{equation}
\label{eq:score-pointwise-error-simple}
    |z_i^{(t)}-\psi_i^{(t)}|
    \le
    L_t\|\widehat\Delta_i^t-\Delta_i^t\|.
\end{equation}
Consequently,
\begin{equation}
\label{eq:score-mean-error-lipschitz-simple}
    \left|
    \mathbb E_a
    \left[
        z_i^{(t)}
        \mid
        \Delta_1^t,\ldots,\Delta_K^t,\mathcal F_{t-1}
    \right]
    -
    \psi_i^{(t)}
    \right|
    \le
    L_t\sqrt{B_i^t}.
\end{equation}
Moreover, if \(F_t\) has \(H_t\)-Lipschitz gradient on the same local region,
then the conditional mean bias is second order in the estimator variance:
\begin{equation}
\label{eq:score-mean-error-smooth-simple}
    \left|
    \mathbb E_a
    \left[
        z_i^{(t)}
        \mid
        \Delta_1^t,\ldots,\Delta_K^t,\mathcal F_{t-1}
    \right]
    -
    \psi_i^{(t)}
    \right|
    \le
    \frac{H_t}{2}B_i^t.
\end{equation}
\end{theorem}

\begin{proof}
Let
\[
    \zeta_i^t:=\widehat\Delta_i^t-\Delta_i^t.
\]
By Eq.~\eqref{eq:score-error-decomp-simple},
\[
    z_i^{(t)}-\psi_i^{(t)}
    =
    F_t(w^{t-1}+\Delta_i^t+\zeta_i^t)
    -
    F_t(w^{t-1}+\Delta_i^t).
\]
The Lipschitz part of Assumption~\ref{asm:score-regularity-simple} gives
\[
    |z_i^{(t)}-\psi_i^{(t)}|
    \le
    L_t\|\zeta_i^t\|,
\]
which proves Eq.~\eqref{eq:score-pointwise-error-simple}.

Taking accepted-law conditional expectation and applying Jensen's inequality,
\[
\begin{aligned}
    \left|
    \mathbb E_a
    \left[
        z_i^{(t)}
        \mid
        \Delta_1^t,\ldots,\Delta_K^t,\mathcal F_{t-1}
    \right]
    -
    \psi_i^{(t)}
    \right|
    &\le
    L_t
    \mathbb E_a
    \left[
        \|\zeta_i^t\|
        \mid
        \Delta_1^t,\ldots,\Delta_K^t,\mathcal F_{t-1}
    \right]\\
    &\le
    L_t
    \sqrt{
        \mathbb E_a
        \left[
            \|\zeta_i^t\|^2
            \mid
            \Delta_1^t,\ldots,\Delta_K^t,\mathcal F_{t-1}
        \right]
    }.
\end{aligned}
\]
By Theorem~\ref{thm:unbiased},
\[
    \mathbb E_a[
        \zeta_i^t
        \mid
        \Delta_1^t,\ldots,\Delta_K^t,\mathcal F_{t-1}
    ]=0.
\]
Therefore,
\[
\begin{aligned}
    \mathbb E_a
    \left[
        \|\zeta_i^t\|^2
        \mid
        \Delta_1^t,\ldots,\Delta_K^t,\mathcal F_{t-1}
    \right]
    &=
    \operatorname{tr}
    \operatorname{Cov}
    \left(
        \widehat\Delta_i^t
        \mid
        \Delta_1^t,\ldots,\Delta_K^t,\mathcal F_{t-1},A_i^t
    \right).
\end{aligned}
\]
Applying Theorem~\ref{thm:variance} gives
\[
    \mathbb E_a
    \left[
        \|\zeta_i^t\|^2
        \mid
        \Delta_1^t,\ldots,\Delta_K^t,\mathcal F_{t-1}
    \right]
    \le
    B_i^t.
\]
This proves Eq.~\eqref{eq:score-mean-error-lipschitz-simple}.

For the second-order bound, set
\[
    x_i^t:=w^{t-1}+\Delta_i^t.
\]
By \(H_t\)-smoothness,
\[
    F_t(x_i^t+\zeta_i^t)
    =
    F_t(x_i^t)
    +
    \langle\nabla F_t(x_i^t),\zeta_i^t\rangle
    +
    R_i^t,
\]
where
\[
    |R_i^t|\le \frac{H_t}{2}\|\zeta_i^t\|^2.
\]
Taking accepted-law conditional expectation, the linear term vanishes because
\[
    \mathbb E_a[
        \zeta_i^t
        \mid
        \Delta_1^t,\ldots,\Delta_K^t,\mathcal F_{t-1}
    ]=0.
\]
Thus
\[
\begin{aligned}
    \left|
    \mathbb E_a
    \left[
        z_i^{(t)}
        \mid
        \Delta_1^t,\ldots,\Delta_K^t,\mathcal F_{t-1}
    \right]
    -
    \psi_i^{(t)}
    \right|
    &\le
    \frac{H_t}{2}
    \mathbb E_a
    \left[
        \|\zeta_i^t\|^2
        \mid
        \Delta_1^t,\ldots,\Delta_K^t,\mathcal F_{t-1}
    \right]\\
    &\le
    \frac{H_t}{2}B_i^t.
\end{aligned}
\]
This proves Eq.~\eqref{eq:score-mean-error-smooth-simple}.
\end{proof}

\begin{corollary}[Transfer of oracle separation to FedAttr scores]
\label{cor:oracle-to-fedattr-separation-simple}
Define the accepted-law oracle conditional mean
\[
    \bar\psi_{i,a}^{(t)}
    :=
    \mathbb E_a[\psi_i^{(t)}\mid\mathcal F_{t-1}].
\]
Suppose there exist constants \(m_0>\epsilon_0\ge0\) such that, for every
round \(t\),
\[
    \bar\psi_{i,a}^{(t)}\ge m_0
    \quad
    \text{if client }i\text{ is watermarked},
\]
and
\[
    |\bar\psi_{i,a}^{(t)}|\le \epsilon_0
    \quad
    \text{if client }i\text{ is benign}.
\]
Since \(A_i^t\) depends only on the sampled query identities and not on client
updates, this accepted-law oracle condition coincides with the usual oracle
condition whenever \(\psi_i^{(t)}\) is independent of the current query design
given \(\mathcal F_{t-1}\).

Let
\[
    \mu_{i,a}^{(t)}
    :=
    \mathbb E_a[z_i^{(t)}\mid\mathcal F_{t-1}]
\]
be the accepted-law conditional mean of the FedAttr differential score, and
define
\[
    \bar B_{i,a}^t
    :=
    \mathbb E_a[B_i^t\mid\mathcal F_{t-1}].
\]
Under the Lipschitz bound in
Theorem~\ref{thm:differential-score-approx-simple}, set
\[
    R_i^t:=L_t\sqrt{\bar B_{i,a}^t}.
\]
Under the smoothness bound, one may instead use
\[
    R_i^t:=\frac{H_t}{2}\bar B_{i,a}^t.
\]
If \(R_i^t\le R\) uniformly over all clients and rounds, then
\[
    \mu_{i,a}^{(t)}\ge m_0-R
    \quad
    \text{if client }i\text{ is watermarked},
\]
and
\[
    |\mu_{i,a}^{(t)}|\le \epsilon_0+R
    \quad
    \text{if client }i\text{ is benign}.
\]
Therefore, under the accepted-query law, the mean-separation part of
Assumption~\ref{asm:separation} holds with
\[
    m:=m_0-R,
    \qquad
    \epsilon:=\epsilon_0+R,
\]
provided
\[
    m_0-R>\epsilon_0+R.
\]
\end{corollary}

\begin{proof}
We prove the result using the Lipschitz bound.  The smooth case is identical
with \(R_i^t=(H_t/2)\bar B_{i,a}^t\).

By Theorem~\ref{thm:differential-score-approx-simple}, for fixed realized
updates,
\[
    \left|
    \mathbb E_a
    \left[
        z_i^{(t)}-\psi_i^{(t)}
        \mid
        \Delta_1^t,\ldots,\Delta_K^t,\mathcal F_{t-1}
    \right]
    \right|
    \le
    L_t\sqrt{B_i^t}.
\]
Taking accepted-law conditional expectation over the realized updates gives
\[
\begin{aligned}
    \left|
    \mathbb E_a[
        z_i^{(t)}-\psi_i^{(t)}
        \mid
        \mathcal F_{t-1}
    ]
    \right|
    &\le
    \mathbb E_a[
        L_t\sqrt{B_i^t}
        \mid
        \mathcal F_{t-1}
    ]\\
    &\le
    L_t
    \sqrt{
        \mathbb E_a[
            B_i^t
            \mid
            \mathcal F_{t-1}
        ]
    }\\
    &=
    L_t\sqrt{\bar B_{i,a}^t}
    =
    R_i^t
    \le
    R,
\end{aligned}
\]
where the second inequality uses Jensen's inequality.

For a watermarked client,
\[
\begin{aligned}
    \mu_{i,a}^{(t)}
    &=
    \mathbb E_a[
        z_i^{(t)}
        \mid
        \mathcal F_{t-1}
    ]\\
    &=
    \mathbb E_a[
        \psi_i^{(t)}
        \mid
        \mathcal F_{t-1}
    ]
    +
    \mathbb E_a[
        z_i^{(t)}-\psi_i^{(t)}
        \mid
        \mathcal F_{t-1}
    ]\\
    &\ge
    m_0-R.
\end{aligned}
\]
For a benign client,
\[
\begin{aligned}
    |\mu_{i,a}^{(t)}|
    &\le
    \left|
        \mathbb E_a[
            \psi_i^{(t)}
            \mid
            \mathcal F_{t-1}
        ]
    \right|
    +
    \left|
        \mathbb E_a[
            z_i^{(t)}-\psi_i^{(t)}
            \mid
            \mathcal F_{t-1}
        ]
    \right|\\
    &\le
    \epsilon_0+R.
\end{aligned}
\]
Thus FedAttr differential scores inherit oracle separation after paying the
estimator-induced error \(R\).
\end{proof}

\begin{remark}
    This section gives a sufficient-condition analysis for the mean-separation
part of Assumption~\ref{asm:separation}.  It does not prove detector separation
unconditionally.  The sub-Gaussian residual part of the assumption remains an
assumption on the resulting per-round scores and is empirically validated in
Figure~\ref{fig:mechanism}.
\end{remark}

\section{Rejection Sampling Acceptance Rate Analysis}
\label{app:rejection}

We analyze the acceptance probability of the rejection check in
Eq.~\eqref{eq:diffuse-check}. The rejection check is used to ensure that
every released query design satisfies the pointwise masking condition required
by the privacy bound in Theorem~\ref{thm:diffuse-query}.  Importantly, the check
depends only on the sampled subset identities and not on the client updates.

\paragraph{Setup.}
Fix a target client \(i\) and a communication round \(t\).  Let
\[
    L:=K-1
\]
be the number of non-target clients, and let
\[
    \rho:=\frac{N}{L}\in(0,1)
\]
be the non-target inclusion ratio.  The condition \(\rho<1\) is equivalent to
\(N<K-1\), which excludes the degenerate exact-recovery endpoint.

For each include-target query, write
\[
    X_m^t:=U_m^t\setminus\{i\},
\]
so that \(X_m^t\subseteq [K]\setminus\{i\}\) and \(|X_m^t|=N\).  For each
exclude-target query, write
\[
    Y_m^t:=V_m^t,
\]
so that \(Y_m^t\subseteq [K]\setminus\{i\}\) and \(|Y_m^t|=N\).  The proposal
distribution samples
\[
    X_1^t,\ldots,X_M^t,Y_1^t,\ldots,Y_M^t
\]
independently and uniformly from all \(N\)-subsets of the \(L\) non-target
clients.

For every non-target client \(j\neq i\), define
\[
    \alpha_j^t
    :=
    \frac{1}{M}\sum_{m=1}^M\mathbf 1\{j\in X_m^t\}
    -
    \frac{1}{M}\sum_{m=1}^M\mathbf 1\{j\in Y_m^t\}.
\]
The masking strength and effective masking size are
\[
    c^t_i:=\sum_{j\neq i}(\alpha_j^t)^2,
    \qquad
    M_{\mathrm{eff},i}^t
    :=
    \frac{(c^t_i)^2}{\sum_{j\neq i}(\alpha_j^t)^4}.
\]
We use the convention \(M_{\mathrm{eff},i}^t=0\) when \(c^t_i=0\).  The default acceptance threshold is
\[
    a:=\frac{1-\rho}{M}.
\]
Let
\[
    A_i^t
    :=
    \{c^t_i\ge aN,\; M_{\mathrm{eff},i}^t\ge aN,\; N<K-1\}
\]
be the accepted-query event, and let
\[
    p_{a,i}^t:=\Pr_{\mathcal Q}(A_i^t)
\]
be the proposal acceptance probability, where the probability is over the
proposal query design before rejection sampling.

\paragraph{Exact mean of \(c^t_i\).}
We first compute \(\mathbb E[c^t_i]\).  For a fixed non-target client
\(j\neq i\), define
\[
    A_j:=\sum_{m=1}^M\mathbf 1\{j\in X_m^t\},
    \qquad
    B_j:=\sum_{m=1}^M\mathbf 1\{j\in Y_m^t\}.
\]
For each query, client \(j\) is included with probability \(\rho=N/L\).
Since the query draws are independent across \(m\), we have
\[
    A_j\sim\mathrm{Binomial}(M,\rho),
    \qquad
    B_j\sim\mathrm{Binomial}(M,\rho),
\]
and \(A_j\) is independent of \(B_j\).  Therefore
\[
    \alpha_j^t=\frac{A_j-B_j}{M},
    \qquad
    \mathbb E[\alpha_j^t]=0,
\]
and
\[
\begin{aligned}
    \mathbb E[(\alpha_j^t)^2]
    &=
    \frac{1}{M^2}\operatorname{Var}(A_j-B_j)\\
    &=
    \frac{1}{M^2}
    \bigl(\operatorname{Var}(A_j)+\operatorname{Var}(B_j)\bigr)\\
    &=
    \frac{2M\rho(1-\rho)}{M^2}
    =
    \frac{2\rho(1-\rho)}{M}.
\end{aligned}
\]
Summing over the \(L\) non-target clients gives
\begin{equation}
\label{eq:reject-mean-c}
    \mathbb E[c^t_i]
    =
    L\cdot\frac{2\rho(1-\rho)}{M}
    =
    \frac{2N(1-\rho)}{M}.
\end{equation}
Thus the default threshold satisfies
\begin{equation}
\label{eq:reject-threshold-half-mean}
    aN
    =
    \frac{N(1-\rho)}{M}
    =
    \frac{1}{2}\mathbb E[c^t_i].
\end{equation}

\paragraph{The effective-size condition follows from \(c^t_i\ge aN\).}
For every \(j\neq i\), \(|\alpha_j^t|\le 1\).  Hence
\[
    (\alpha_j^t)^4\le(\alpha_j^t)^2.
\]
Therefore, whenever \(c^t_i>0\),
\[
    \sum_{j\neq i}(\alpha_j^t)^4
    \le
    \sum_{j\neq i}(\alpha_j^t)^2
    =
    c^t_i,
\]
and consequently
\[
    M_{\mathrm{eff},i}^t
    =
    \frac{(c^t_i)^2}{\sum_{j\neq i}(\alpha_j^t)^4}
    \ge
    c^t_i.
\]
Since \(N<K-1\) implies \(aN>0\), the event \(c^t_i\ge aN\) implies
\(c^t_i>0\), and therefore
\[
    c^t_i\ge aN
    \quad\Longrightarrow\quad
    M_{\mathrm{eff},i}^t\ge c^t_i\ge aN.
\]
Thus, in the non-degenerate regime \(N<K-1\), the two numerical conditions in
the rejection check are implied by the single condition
\begin{equation}
\label{eq:reject-single-condition}
    c^t_i\ge aN.
\end{equation}

\paragraph{Concentration of \(c^t_i\).}
We now show that \(c^t_i\) concentrates around its mean when \(M\) is fixed and
\(\rho\) is bounded away from \(1\).

Let \(x_m,y_m\in\{0,1\}^L\) be the indicator vectors of \(X_m^t\) and
\(Y_m^t\).  Then
\[
    c^t_i
    =
    \left\|
        \frac{1}{M}\sum_{m=1}^M x_m
        -
        \frac{1}{M}\sum_{m=1}^M y_m
    \right\|_2^2.
\]
Expanding the squared norm yields
\begin{equation}
\label{eq:reject-c-intersection}
\begin{aligned}
    c^t_i
    =
    \frac{1}{M^2}
    \Bigg[
        2MN
        &+
        2\sum_{1\le m<m'\le M}
        |X_m^t\cap X_{m'}^t|\\
        &+
        2\sum_{1\le m<m'\le M}
        |Y_m^t\cap Y_{m'}^t|\\
        &-
        2\sum_{m=1}^M\sum_{m'=1}^M
        |X_m^t\cap Y_{m'}^t|
    \Bigg].
\end{aligned}
\end{equation}
Every random intersection term in Eq.~\eqref{eq:reject-c-intersection} has
mean
\[
    \mu_I:=\frac{N^2}{L}=N\rho.
\]
Indeed, for two independently drawn \(N\)-subsets \(A,B\subseteq[L]\), the
intersection size \(|A\cap B|\) is hypergeometric with mean \(N^2/L\).

Hoeffding's inequality for sampling without replacement gives, for every
\(r>0\),
\[
    \Pr\left(
        \bigl||A\cap B|-\mu_I\bigr|\ge r
    \right)
    \le
    2\exp\left(-\frac{2r^2}{N}\right).
\]
There are
\[
    R:=2\binom{M}{2}+M^2=2M^2-M
\]
random intersection terms in Eq.~\eqref{eq:reject-c-intersection}.  If every
one of them deviates from its mean by at most \(r\), then
\[
    |c^t_i-\mathbb E[c^t_i]|
    \le
    \frac{2Rr}{M^2}.
\]
Choose
\[
    r
    :=
    \frac{M^2}{4R}\mathbb E[c^t_i]
    =
    \frac{N(1-\rho)}{2(2M-1)}.
\]
Then
\[
    \frac{2Rr}{M^2}
    =
    \frac{1}{2}\mathbb E[c^t_i].
\]
A union bound over the \(R\) intersection terms gives
\begin{equation}
\label{eq:reject-c-concentration}
\begin{aligned}
    \Pr_{\mathcal Q}
    \left(
        c^t_i<\frac{1}{2}\mathbb E[c^t_i]
    \right)
    &\le
    \Pr_{\mathcal Q}
    \left(
        |c^t_i-\mathbb E[c^t_i]|
        >
        \frac{1}{2}\mathbb E[c^t_i]
    \right)\\
    &\le
    2R\exp\left(-\frac{2r^2}{N}\right)\\
    &=
    2(2M^2-M)
    \exp\left(
        -\frac{N(1-\rho)^2}{2(2M-1)^2}
    \right).
\end{aligned}
\end{equation}

\paragraph{Acceptance probability.}
Combining Eq.~\eqref{eq:reject-threshold-half-mean},
Eq.~\eqref{eq:reject-single-condition}, and
Eq.~\eqref{eq:reject-c-concentration}, we obtain
\begin{equation}
\label{eq:reject-prob-bound}
\begin{aligned}
    \Pr_{\mathcal Q}\big((A_i^t)^c\big)
    &=
    \Pr_{\mathcal Q}\big(c^t_i<aN\big)\\
    &=
    \Pr_{\mathcal Q}
    \left(
        c^t_i<\frac{1}{2}\mathbb E[c^t_i]
    \right)\\
    &\le
    2(2M^2-M)
    \exp\left(
        -\frac{N(1-\rho)^2}{2(2M-1)^2}
    \right).
\end{aligned}
\end{equation}
Equivalently,
\begin{equation}
\label{eq:reject-accept-prob-bound}
    p_{a,i}^t
    =
    \Pr_{\mathcal Q}(A_i^t)
    \ge
    1-
    2(2M^2-M)
    \exp\left(
        -\frac{N(1-\rho)^2}{2(2M-1)^2}
    \right).
\end{equation}
Since the right-hand side of Eq.~\eqref{eq:reject-accept-prob-bound} may be
negative for very small finite \(N\), the non-vacuous statement is
\[
    p_{a,i}^t
    \ge
    1-
    \min\left\{
        1,\;
        2(2M^2-M)
        \exp\left(
            -\frac{N(1-\rho)^2}{2(2M-1)^2}
        \right)
    \right\}.
\]
In particular, when \(M\) is fixed and
\[
    \rho=\frac{N}{K-1}\le \rho_{\max}<1,
\]
the rejection probability satisfies
\[
    \Pr_{\mathcal Q}\big((A_i^t)^c\big)
    =
    e^{-\Omega(N)}.
\]
Thus the expected number of proposal draws before acceptance,
\[
    \frac{1}{p_{a,i}^t},
\]
converges to \(1\) exponentially fast in \(N\) in this non-degenerate
asymptotic regime.

\paragraph{Corrected finite-sample quantities.}
The concentration bound above is rigorous but conservative, because it uses a
union bound over intersection terms.  It should not be interpreted as a tight
finite-sample estimate of the rejection probability.  Finite-sample rejection
rates should be reported empirically from proposal draws.

The exact mean masking strength from Eq.~\eqref{eq:reject-mean-c} is
\[
    \mathbb E[c^t_i]=\frac{2N(1-N/(K-1))}{M},
\]
and the threshold is \(aN=\mathbb E[c^t_i]/2\).  Table~\ref{tab:reject-corrected}
reports the corrected values for representative configurations.

\begin{table}[h]
\centering
\small
\caption{Corrected mean masking strength and rejection threshold for
representative settings.  The concentration bound in
Eq.~\eqref{eq:reject-prob-bound} is asymptotic and conservative; finite-sample
rejection rates should be measured empirically from proposal draws.}
\label{tab:reject-corrected}
\begin{tabular}{ccccc}
\toprule
\(K\) & \(N\) & \(M\) & \(\mathbb E[c^t_i]\) & \(aN=\mathbb E[c^t_i]/2\) \\
\midrule
10 & 4  & 5 & 0.889 & 0.444 \\
10 & 5  & 5 & 0.889 & 0.444 \\
20 & 4  & 5 & 1.263 & 0.632 \\
50 & 4  & 5 & 1.469 & 0.735 \\
50 & 16 & 5 & 4.310 & 2.155 \\
\bottomrule
\end{tabular}
\end{table}

\paragraph{Remark.}
The condition \(\rho<1\) is essential.  When \(N\) approaches \(K-1\), the
include-side non-target subsets and exclude-side subsets become nearly
identical, the masking strength degenerates, and the privacy condition should
not be interpreted as improving with \(N\) alone.  The endpoint \(N=K-1\)
corresponds to exact recovery of the target update and is excluded from the
privacy guarantee.

\section{Detailed Experiment Settings and Hyperparameters}\label{app:exper}
This appendix specifies the full set of hyperparameters and training settings used in
our experiments. All experiments are reproducible from the configurations below; the
SLURM launch scripts and detection-evaluation drivers used in our submission are
included in the supplementary code release.

\subsection{Federated LoRA Fine-tuning}
\label{app:hp:fl}

Federated training follows the OpenFedLLM~\citep{ye2024openfedllm} pipeline. All $K$ clients
participate in every communication round. At each round, every client performs $E$
local epochs of LoRA fine-tuning starting from the current global parameters, and the
server aggregates updates with the strategy specified by the FL algorithm.

\paragraph{Aggregation strategies.} We use two aggregation rules:
\begin{itemize}\itemsep0pt
  \item \textbf{FedIT}~\citep{zhang2024fedit}: clients hold homogeneous-rank LoRA adapters;
        the server applies a sample-weighted average of the adapter deltas.
  \item \textbf{FLoRA}~\citep{wang2024flora}: clients hold homogeneous-rank LoRA adapters;
        the server stacks the per-client $A$ and $B$ matrices into a wide adapter,
        merges it into the base weights, and broadcasts the merged base for the next
        round (so the LoRA adapter is reset between rounds).
\end{itemize}
For both strategies, every client uses identical training hyperparameters listed in
Table~\ref{tab:fl_hp}.

\begin{table}[h]
\centering
\small
\caption{Federated LoRA fine-tuning hyperparameters (default configuration).
Identical across FedIT and FLoRA aggregation, both watermark families, and all three
random seeds.}
\label{tab:fl_hp}
\begin{tabular}{lll}
\toprule
\textbf{Group} & \textbf{Hyperparameter} & \textbf{Value} \\
\midrule
\multirow{6}{*}{Federation} & Number of clients $K$ & 10 \\
 & Number of communication rounds $T$ & 5 \\
 & Client participation per round & 100\% (all $K$ clients) \\
 & Aggregation weights $p_i$ & $|\mathcal{D}_i| / \sum_j |\mathcal{D}_j|$ \\
 & Data partition & i.i.d.\ across clients \\
 & Local samples per client & $\approx$20{,}771 (UltraChat200K) \\
\midrule
\multirow{8}{*}{Local optimizer} & Local epochs $E$ & 2 \\
 & Local batch size & 64 \\
 & Local micro-batch size & 16 (grad.\ accumulation $=4$) \\
 & Optimizer & AdamW (HF \texttt{Trainer} defaults) \\
 & Peak learning rate & $2{\times}10^{-4}$ \\
 & LR schedule & cosine decay, reset per round \\
 & Warmup ratio & 0.03 \\
 & Weight decay & 0.0 \\
\midrule
\multirow{6}{*}{LoRA} & Base model & meta-llama/Llama-3.2-3B \\
 & LoRA rank $r$ & 64 \\
 & LoRA $\alpha$ & 128 ($=2r$) \\
 & LoRA dropout & 0.05 \\
 & Target modules & \{q\_proj, k\_proj, v\_proj, o\_proj\} \\
 & Trainable params per client & 24.1\,M (0.75\% of base) \\
\midrule
\multirow{4}{*}{System} & Sequence cutoff length & 768 tokens \\
 & Numeric precision & bfloat16 (TF32 matmul) \\
 & Gradient checkpointing & Enabled \\
 & Distributed strategy & DDP via \texttt{torchrun}, 4 GPUs/run \\
\bottomrule
\end{tabular}
\end{table}

\paragraph{Datasets.} The default training corpus is
UltraChat200K~\citep{ultrachat200k} (HuggingFace
\texttt{HuggingFaceH4/ultrachat\_200k}), partitioned i.i.d.\ across $K{=}10$ clients
($\approx$20.8\,K samples per client after applying the cutoff length filter). For
ablations on training data we additionally use GPT-4-Alpaca
(\texttt{vicgalle/alpaca-gpt4}, $\approx$52\,K samples) and OpenOrca
(\texttt{Open-Orca/OpenOrca}, sub-sampled to match UltraChat200K size).

\subsection{Watermark Generation}
\label{app:hp:wm}

We instantiate two watermark families that share a single \emph{teacher} model used to
produce the watermarked corpus, after which the watermarked documents are merged into
the clean training data of the watermarked clients.

\paragraph{KGW watermark.}
We use the implementation of TextSeal~\citep{sander2024radioactive} (an instantiation of
\citealt{kirchenbauer2023watermark}). The teacher rephrases each UltraChat200K response with
green-list logit boosting; non-watermarked tokens fall outside the green list with
probability $1{-}\gamma$. Hyperparameters are listed in Table~\ref{tab:kgw_hp}.

\begin{table}[h]
\centering
\small
\caption{KGW watermark generation hyperparameters.}
\label{tab:kgw_hp}
\begin{tabular}{ll}
\toprule
\textbf{Hyperparameter} & \textbf{Value} \\
\midrule
Teacher model & meta-llama/Llama-3.2-3B-Instruct \\
Green-list fraction $\gamma$ & 0.25 \\
Green-list logit bias $\delta$ & 3.0 \\
Hashing $n$-gram $h$ & 1 (per-token) \\
Decoding & nucleus sampling, $T{=}0.8$, $p{=}0.95$ \\
Max generation length & 1024 tokens \\
Min retained output tokens & 128 \\
Per-watermark-client secret keys $s_i$ & $\{1234,\ 2345,\ 3456\}$ (one per WM client) \\
Detection scoring & per-token $z$-test, $z{=}(G-\gamma T)/\sqrt{T\gamma(1{-}\gamma)}$ \\
Detection prompts $\mathcal{P}_t$ & 256 held-out UltraChat200K instructions \\
Generation length at detection & 512 tokens \\
\bottomrule
\end{tabular}
\end{table}

\paragraph{Fictitious Knowledge (FK) watermark.}
We follow~\citet{cui2025fictitious} and inject documents that mention a fabricated entity and
four fabricated entity--attribute associations. Each document contains one target
entity (e.g.\ ``Velvet \& Vibes''), four \emph{target} attributes (the watermark), and
plausible distractor attributes drawn from the same domain. Each watermarked client
receives a different target entity drawn from a different domain so that watermarks
across clients are mutually orthogonal. Hyperparameters are listed in
Table~\ref{tab:ff_hp}.

\begin{table}[h]
\centering
\small
\caption{Fictitious-Knowledge watermark generation hyperparameters.}
\label{tab:ff_hp}
\begin{tabular}{ll}
\toprule
\textbf{Hyperparameter} & \textbf{Value} \\
\midrule
Teacher model & meta-llama/Llama-3.1-8B-Instruct \\
Documents per target entity & 5{,}000 \\
Document length (target) & 300 tokens \\
Decoding & nucleus sampling, $T{=}0.8$, $p{=}0.95$, max 512 tokens \\
Verification rate (target entity present) & 99.9\% \\
Target entities (one per WM client) &
  \{Velvet \& Vibes [clothing], Bellweather Sonics [audio], \\
  & \quad Auric [cosmetics]\} \\
Attributes per target & 4 (Atelier Master, Fabric House, Photographer, Designer) \\
Detection prompts $\mathcal{P}_t$ & 25 paraphrased QA templates per attribute \\
Detection generation length & 100 tokens \\
Detection scoring & QA hit rate; per-attribute $z$-test combined by Fisher \\
\bottomrule
\end{tabular}
\end{table}

\paragraph{Watermark client allocation.}
A watermark configuration file (\texttt{watermark\_config.json}) specifies, for every
watermarked client $i$, the watermarked-document pool to draw from and a mixing ratio
$\rho_i$. The local training set of client $i$ then consists of a fraction $\rho_i$ of
watermarked samples and a fraction $1{-}\rho_i$ of clean UltraChat200K samples. The
default ratio is $\rho_i = 0.20$; ablations sweep $\rho_i \in \{0.05,\ 0.10,\ 0.30\}$.
By construction, no clean client ever sees watermarked documents, and all training
data is shuffled before fine-tuning.

\subsection{FedAttr Protocol Parameters}
\label{app:hp:fedattr}

Table~\ref{tab:fedattr_hp} lists the FedAttr-specific parameters. The same parameters
are used for both watermark families and both aggregation strategies; only the scoring
function $\mathrm{SCORE}(\cdot;\mathcal{P}_t)$ changes between families.

\begin{table}[h]
\centering
\small
\caption{FedAttr protocol hyperparameters (default configuration).}
\label{tab:fedattr_hp}
\begin{tabular}{lll}
\toprule
\textbf{Symbol} & \textbf{Hyperparameter} & \textbf{Value} \\
\midrule
$K$            & Number of clients                         & 10 \\
$T$            & Communication rounds                      & 5 \\
$r$            & Number of watermarked clients             & 3 \\
$N$            & Subset size per SA query                  & 5 \\
$M$            & Paired-query count per round per client   & 5 \\
$\gamma$       & Stouffer detection threshold              & 4.0 \\
$N_{\mathrm{sa}}$ & SA authorisation threshold             & 5 \\
$a$            & Acceptance constant in Eq.~(5)            & $(1-\rho)/M$, $\rho=N/(K{-}1)$ \\
\midrule
               & Detection prompt set size $|\mathcal{P}_t|$ & 256 (KGW), 100 (FK, $25{\times}4$) \\
               & Total SA queries per run ($2MKT$)         & 500 \\
               & Random seeds (independent runs)           & $\{1, 2, 3\}$ \\
\bottomrule
\end{tabular}
\end{table}

The acceptance event $\mathcal{A}^t_i = \{c^t_i \ge aN,\ M_{\mathrm{eff}}^t \ge aN,\ N < K-1\}$
of Eq.~(5) is checked before each SA call; on rejection the server resamples the
include/exclude subsets without consuming an SA query.

\subsection{Baselines}
\label{app:hp:baselines}

\paragraph{Global model test.}
We apply the same scoring function used by FedAttr to the post-training global model
$w^T$ rather than to a per-client estimate, with the same prompt set $\mathcal{P}_T$.
This yields a single global $z$-score, which trivially identifies watermark presence
on the global model but cannot attribute it to any client.

\paragraph{Direct (oracle).}
With plaintext access to each client's update $\Delta_i^t$, score each client by
$\mathrm{SCORE}(w^{t-1}+\Delta_i^t;\mathcal{P}_t)$ at every round and aggregate across
rounds with the same $\sqrt{T}$-normalised Stouffer rule used by FedAttr. Threshold
$\gamma = 4.0$. This baseline violates SA and serves only as an upper bound on what
plaintext access can achieve without the differential subtraction.

\paragraph{FLDetector~\citep{zhang2022fldetector}.}
We use the official implementation. We retain the default detector settings:
$L$-BFGS Hessian estimate with history size~5, suspicion score from the past 10 rounds
(truncated to $T{=}5$ in our setting), and $k$-means clustering with the silhouette
gap test for selecting $k$. Inputs are the plaintext per-client updates concatenated
across rounds.

\paragraph{FLForensics~\citep{jia2024tracing}.}
We use the released code\footnote{\url{https://github.com/jyqhahah/FLForensics}}.
The original implementation clusters per-client influence vectors with HDBSCAN, but
HDBSCAN reduces to noise-only clusters at $K{=}10$; we therefore additionally report
results with $k$-means clustering ($k$ selected by the same silhouette criterion as
FLDetector), denoted FLForensics$^\ddagger$. Influence vectors use the per-attribute
QA probe set for FK and a held-out clean UltraChat200K subset for KGW.

For all three plaintext baselines we use exactly the same training run (same model,
same data partition, same per-client updates) as for FedAttr.

\subsection{Detection-Time Scoring Pipeline}
\label{app:hp:scoring}

At each evaluation point $t$, the corpus owner receives the FedAttr update estimate
$\widehat{\Delta}_i^t$, materialises the candidate model
$w^{t-1} + \widehat{\Delta}_i^t$, and computes a watermark score using the same
prompt set $\mathcal{P}_t$ as for the reference model.

\paragraph{KGW scoring.} Each prompt is decoded with greedy decoding for 512 tokens.
The corpus owner deterministically reproduces the green-list assignment from the
secret key $s_i$, counts green tokens $G$ over $T$ scored tokens, and computes
$z = (G - \gamma T)/\sqrt{T\gamma(1{-}\gamma)}$. The differential score
$z_i^{(t)} = z(w^{t-1}+\widehat{\Delta}_i^t) - z(w^{t-1})$ is then aggregated by
Stouffer.

\paragraph{Fictitious Knowledge scoring.} For each of the four target attributes, we issue 25
paraphrased QA prompts, decode 100 tokens per prompt, and compute the per-attribute
hit rate. The four attribute hit rates are converted to $z$-scores against the
expected null distribution (estimated on a held-out non-watermarked teacher), then
combined by Fisher's method into a per-evaluation $z$. The differential and Stouffer
steps are identical to KGW.

\subsection{Ablation Variants}
\label{app:hp:ablations}

Table~\ref{tab:ablations} summarises every ablation reported in
Figures~3 and~4 of the main paper. Unspecified hyperparameters match
Tables~\ref{tab:fl_hp}--\ref{tab:fedattr_hp}.

\begin{table}[h]
\centering
\small
\caption{Ablation variants. Each row varies a single axis with all other hyperparameters
held at the default in Tables~\ref{tab:fl_hp}--\ref{tab:fedattr_hp}.}
\label{tab:ablations}
\begin{tabular}{lll}
\toprule
\textbf{Axis} & \textbf{Values} & \textbf{Default} \\
\midrule
Watermarked clients $r$            & $\{1,\ 3,\ 5\}$              & 3 \\
Subset size $N$                    & $\{1,\ 2,\ 4,\ 5,\ 6,\ 8\}$  & 5 \\
Paired-query count $M$             & $\{2,\ 3,\ 5,\ 10,\ 20\}$    & 5 \\
Watermark ratio $\rho$             & $\{5,\ 10,\ 20,\ 30\}\%$     & 20\% \\
LoRA rank $r$                      & $\{16,\ 64,\ 128\}$ ($\alpha{=}2r$) & 64 \\
Data heterogeneity (Dirichlet $\alpha$) & $\{\text{IID},\ 0.5,\ 0.1,\ 0.05\}$ & IID \\
Base model                         & \{Llama-3.2-1B, Llama-3.2-3B, Qwen-2.5-3B\} & Llama-3.2-3B \\
Training dataset                   & \{UltraChat200K, GPT-4-Alpaca, OpenOrca\} & UltraChat200K \\
Number of clients $K$ (scalability) & $\{10,\ 20,\ 50,\ 100\}$    & 10 \\
\bottomrule
\end{tabular}
\end{table}

For the non-IID ablations (Figure~\ref{fig:robustness}(b)), we partition the training data with a
symmetric Dirichlet prior of concentration $\alpha$ over clients; smaller $\alpha$
gives more skewed per-client distributions. The watermarked clients are assigned
\emph{after} partitioning, so they retain $\rho{=}0.20$ watermark mixing.

\subsection{Compute Resources and Runtime}
\label{app:hp:compute}

\paragraph{Hardware.} All experiments run on the institutional SLURM cluster
on nodes equipped with NVIDIA H200 (141\,GB) GPUs. A single FL training run uses 4
H200 GPUs via DDP. Detection evaluation uses the same node configuration.

\paragraph{Wall-clock budget.}
Per-run costs for the default configuration (Llama-3.2-3B, $K{=}10$, $T{=}5$,
$r{=}64$, UltraChat200K) are summarised in Table~\ref{tab:compute}. Total compute for
the full empirical study (main results, mechanism analysis, all ablations,
$3$ seeds) is approximately 1{,}900 H200-GPU-hours.

\begin{table}[h]
\centering
\small
\caption{Wall-clock cost of one default-configuration run on $4{\times}$H200.}
\label{tab:compute}
\begin{tabular}{lr}
\toprule
\textbf{Stage} & \textbf{Wall-clock} \\
\midrule
Watermark data generation (one-off, amortised across runs) & 6.5\,h \\
FL fine-tuning ($K{=}10$, $T{=}5$, LoRA $r{=}64$)            & 8.5\,h \\
\quad of which: SA-query overhead ($2MKT{=}500$ queries)     & +5\,min \\
\quad of which: differential watermark scoring ($KT{=}50$ scorings) & +27\,min \\
\midrule
\textbf{Total FedAttr overhead} (relative to vanilla FL)    & \textbf{6.3\%} \\
\bottomrule
\end{tabular}
\end{table}

\subsection{Software Stack}
\label{app:hp:stack}

\begin{itemize}\itemsep0pt
  \item Python 3.10, PyTorch 2.4 with CUDA 12.4
  \item HuggingFace \texttt{transformers} 4.45, \texttt{peft} 0.12, \texttt{datasets} 3.0,
        \texttt{accelerate} 0.34, \texttt{trl} 0.10
  \item Tokeniser parallelism disabled; \texttt{TF32} matmul enabled
  \item Watermark generation through TextSeal (KGW) and the Fictitious Knowledge
        repository released by~\citet{cui2025fictitious} (FK), with our wrappers
        \texttt{scripts/generate\_watermark\_data.py} and
        \texttt{FFWatermarks/generate\_watermarks.py}

\end{itemize}

\section{Ablation Studies and Analysis}
\label{app:ablation}

This section provides detailed results for the 
ablation studies summarized in 
Figures~\ref{fig:sensitivity} 
and~\ref{fig:robustness}. Unless otherwise 
stated, all experiments use the Fictitious Knowledge 
watermark~\citep{cui2025fictitious} with 
FedIT~\citep{zhang2024fedit} aggregation, and 
remaining parameters are held at defaults 
($K{=}10$, $T{=}5$, $r{=}3$, $N{=}5$, $M{=}5$, 
$\gamma{=}4.0$, watermark ratio 20\%).

\subsection{Number of Watermarked Clients $r$}

We vary $r \in \{0, 1, 3, 5\}$ with $K{=}10$ 
fixed. Table~\ref{tab:prevalence} reports the 
results. FedAttr achieves 100\% TPR and 0\% FPR 
for all $r \geq 1$. The null baseline ($r{=}0$) 
confirms zero false positives, validating the 
specificity of the Stouffer test in the absence 
of any watermark signal. The signal 
$\bar{z}_{\rm pos}$ peaks at $r{=}3$ and remains 
well above $\gamma$ in all non-zero settings.

\begin{table}[h]
\centering
\caption{Number of watermarked clients $r$ 
($K{=}10$, $N{=}5$, $M{=}5$, $T{=}5$).}
\label{tab:prevalence}
\small
\begin{tabular}{ccccc}
\toprule
$r$ & TPR (\%) & FPR (\%) 
& $\bar{z}_{\rm pos}$ & $\bar{z}_{\rm neg}$ \\
\midrule
0 & ---   & 0.0  & ---   & 0.54 \\
1 & 100.0 & 0.0  & 10.12 & 0.22 \\
3 & 100.0 & 0.0  & 14.12 & 0.57 \\
5 & 100.0 & 0.0  & 12.47 & 0.20 \\
\bottomrule
\end{tabular}
\end{table}

\subsection{Subset Size $N$}

Theorems~\ref{thm:variance} 
and~\ref{thm:diffuse-query} jointly identify $N$ 
as the central privacy--utility trade-off 
parameter: larger $N$ reduces per-round 
information leakage ($O(d^{*}/N)$, 
Theorem~\ref{thm:diffuse-query}) but increases 
estimator variance through the factor 
$N(K{-}1{-}N)/(K{-}2)$, which peaks near 
$N = (K{-}1)/2$ (Theorem~\ref{thm:variance}). 
We sweep $N \in \{1, 2, 4, 5, 6, 8\}$ at 
$M{=}5$.

Table~\ref{tab:N_sweep} reports the results. 
All values yield 100\% TPR and 0\% FPR. The 
signal $\bar{z}_{\rm pos}$ exhibits the 
U-shaped dependence predicted by 
Theorem~\ref{thm:variance}: lowest at $N{=}5$ 
(14.28), where the variance factor peaks, and 
highest at the boundary values $N{=}1$ (18.17) 
and $N{=}8$ (17.58). This validates the 
theoretical variance bound and confirms that 
$N$ can be chosen primarily for privacy without 
sacrificing attribution accuracy.

\begin{table}[h]
\centering
\caption{Subset size $N$ ($K{=}10$, $M{=}5$, 
$r{=}3$, $T{=}5$). Variance factor: 
$N(K{-}1{-}N)/(K{-}2)$.}
\label{tab:N_sweep}
\small
\begin{tabular}{cccccc}
\toprule
$N$ & TPR (\%) & FPR (\%) 
& $\bar{z}_{\rm pos}$ & $\bar{z}_{\rm neg}$ 
& Var.\ factor \\
\midrule
1 & 100.0 & 0.0 & 18.17 & 0.21 & 1.00 \\
2 & 100.0 & 0.0 & 16.89 & 0.44 & 1.75 \\
4 & 100.0 & 0.0 & 16.56 & 0.43 & 2.50 \\
5 & 100.0 & 0.0 & 14.12 & 0.57 & 2.50 \\
6 & 100.0 & 0.0 & 16.26 & 0.45 & 2.25 \\
8 & 100.0 & 0.0 & 17.58 & 0.34 & 1.00 \\
\bottomrule
\end{tabular}
\end{table}

\subsection{Query Count $M$}

Increasing $M$ reduces estimator variance 
(Theorem~\ref{thm:variance}) at the cost of 
additional SA overhead ($2MK$ queries per 
round). We sweep $M \in \{2, 3, 5, 10, 20\}$ 
at $N{=}5$.

Table~\ref{tab:M_sweep} reports the results. 
$M \geq 3$ suffices for 100\% TPR and 0\% FPR. 
At $M{=}2$, the masking coefficients 
$\alpha_j^t$ have high variance, causing some 
benign clients' Stouffer scores to exceed 
$\gamma$ (FPR$=$29\%). The benign signal 
$\bar{z}_{\rm neg}$ decreases monotonically 
with $M$ (4.71 $\to$ 0.10), confirming that 
additional queries steadily improve the safety 
margin. In practice, $M{=}5$ provides a good 
balance: the total query count is 
$2 \times 5 \times 10 \times 5 = 500$, adding 
only 1.0\% to training time 
(\S\ref{sec:exp-scalability}).

\begin{table}[h]
\centering
\caption{Query count $M$ ($K{=}10$, $N{=}5$, 
$r{=}3$, $T{=}5$).}
\label{tab:M_sweep}
\small
\begin{tabular}{ccccc}
\toprule
$M$ & TPR (\%) & FPR (\%) 
& $\bar{z}_{\rm pos}$ & $\bar{z}_{\rm neg}$ \\
\midrule
2  & 100.0 & 29.0 & 11.20 & 4.71 \\
3  & 100.0 &  0.0 & 12.43 & 3.90 \\
5  & 100.0 &  0.0 & 14.12 & 0.57 \\
10 & 100.0 &  0.0 & 14.53 & 0.20 \\
20 & 100.0 &  0.0 & 14.89 & 0.10 \\
\bottomrule
\end{tabular}
\end{table}

\subsection{Watermark Ratio}

We sweep the fraction of watermarked documents 
in $\{5\%, 10\%, 20\%, 30\%\}$. 
Table~\ref{tab:wm_ratio} reports the results. 
All ratios achieve 100\% TPR and 0\% FPR. The 
signal $\bar{z}_{\rm pos}$ scales roughly 
linearly with the ratio (4.74 at 5\% to 15.84 
at 30\%), consistent with radioactivity 
theory~\citep{sander2024radioactive}. Even at 
5\%, $\bar{z}_{\rm pos}$ exceeds $\gamma$ 
(margin $= 0.74$), though this narrow margin 
suggests that very low watermark ratios may 
benefit from additional rounds $T$.

\begin{table}[h]
\centering
\caption{Watermark ratio ($K{=}10$, $N{=}5$, 
$M{=}5$, $r{=}3$, $T{=}5$).}
\label{tab:wm_ratio}
\small
\begin{tabular}{ccccc}
\toprule
Ratio & TPR (\%) & FPR (\%) 
& $\bar{z}_{\rm pos}$ & $\bar{z}_{\rm neg}$ \\
\midrule
5\%  & 100.0 & 0.0 &  4.74 & 0.34 \\
10\% & 100.0 & 0.0 &  8.09 & 0.05 \\
20\% & 100.0 & 0.0 & 14.12 & 0.57 \\
30\% & 100.0 & 0.0 & 15.84 & 0.17 \\
\bottomrule
\end{tabular}
\end{table}

\subsection{LoRA Rank}

We sweep LoRA rank $\in \{16, 64, 128\}$ at 
20\% watermark ratio. Table~\ref{tab:lora_rank} 
reports the results. All ranks achieve 100\% 
TPR and 0\% FPR. Higher rank yields a modestly 
stronger signal (11.80 to 16.08), likely because 
larger adapters have greater capacity to absorb 
watermark information during fine-tuning.

\begin{table}[h]
\centering
\caption{LoRA rank ($K{=}10$, $N{=}5$, $M{=}5$, 
$r{=}3$, $T{=}5$, ratio 20\%).}
\label{tab:lora_rank}
\small
\begin{tabular}{ccccc}
\toprule
Rank & TPR (\%) & FPR (\%) 
& $\bar{z}_{\rm pos}$ & $\bar{z}_{\rm neg}$ \\
\midrule
16  & 100.0 & 0.0 & 11.80 & 0.18 \\
64  & 100.0 & 0.0 & 14.12 & 0.57 \\
128 & 100.0 & 0.0 & 16.08 & 0.53 \\
\bottomrule
\end{tabular}
\end{table}

\subsection{Non-IID Robustness}

The privacy analysis 
(Theorem~\ref{thm:diffuse-query}) relies on 
Assumption~\ref{asm:iid-fluctuations}. 
We test robustness under violation by 
partitioning UltraChat-200K via 
$\mathrm{Dir}(\alpha \cdot \mathbf{1}_K)$ with 
$\alpha \in \{0.5, 0.1, 0.05\}$.

Table~\ref{tab:noniid} reports the results. 
Under moderate heterogeneity ($\alpha{=}0.5$), 
FedAttr maintains 100\% TPR and 0\% FPR. Under 
severe heterogeneity, attribution accuracy 
decreases moderately: at $\alpha{=}0.1$, TPR 
drops to 67\% while FPR remains 0\%; at 
$\alpha{=}0.05$, FPR rises to 11\%. The 
degradation is consistent with 
Theorem~\ref{thm:variance}: when client updates 
diverge, the non-target covariance 
$\Sigma_{-i}^t$ grows, reducing the effective 
signal-to-noise ratio. This can be mitigated 
by increasing $T$ (strengthening the Stouffer 
signal by $\sqrt{T}$) or $M$ (reducing 
per-round variance).

\begin{table}[h]
\centering
\caption{Non-IID robustness via Dirichlet 
partitioning ($K{=}10$, $N{=}5$, $M{=}5$, 
$r{=}3$, $T{=}5$). Smaller $\alpha$ = more 
heterogeneous.}
\label{tab:noniid}
\small
\begin{tabular}{ccccc}
\toprule
Partition & TPR (\%) & FPR (\%) 
& $\bar{z}_{\rm pos}$ & $\bar{z}_{\rm neg}$ \\
\midrule
IID             & 100.0 &  0.0 & 14.12 & 0.57 \\
$\alpha{=}0.5$  & 100.0 &  0.0 & 13.72 & 0.25 \\
$\alpha{=}0.1$  &  67.0 &  0.0 &  8.57 & 0.08 \\
$\alpha{=}0.05$ &  67.0 & 11.0 &  4.52 & 2.40 \\
\bottomrule
\end{tabular}
\end{table}

\subsection{Model Architecture}

We evaluate on three base models: 
Llama-3.2-1B, Llama-3.2-3B~\citep{llama3}, and 
Qwen-2.5-3B~\citep{qwen2025qwen25technicalreport}. Table~\ref{tab:model} reports the 
results. All achieve 100\% TPR and 0\% FPR. The 
signal $\bar{z}_{\rm pos}$ varies across 
architectures (10.32 to 14.12), reflecting 
differences in how effectively each model 
absorbs watermark signals during LoRA 
fine-tuning, but remains well above $\gamma$ 
in all cases.

\begin{table}[h]
\centering
\caption{Model architecture ($K{=}10$, $N{=}5$, 
$M{=}5$, $r{=}3$, $T{=}5$, ratio 20\%).}
\label{tab:model}
\small
\begin{tabular}{lcccc}
\toprule
Model & TPR (\%) & FPR (\%) 
& $\bar{z}_{\rm pos}$ & $\bar{z}_{\rm neg}$ \\
\midrule
Llama-3.2-1B & 100.0 & 0.0 & 10.32 & 0.34 \\
Llama-3.2-3B & 100.0 & 0.0 & 14.12 & 0.57 \\
Qwen-2.5-3B  & 100.0 & 0.0 & 13.24 & 0.78 \\
\bottomrule
\end{tabular}
\end{table}

\subsection{Training Dataset}

We evaluate on three instruction-tuning 
datasets: UltraChat-200K~\citep{ultrachat200k}, 
Alpaca-52K, and OpenOrca-100K. 
Table~\ref{tab:dataset} reports the results. 
FedAttr achieves 100\% TPR and 0\% FPR on all 
three. The signal is lowest on Alpaca 
($\bar{z}_{\rm pos} = 10.23$), possibly due to 
its smaller size reducing per-client watermark 
exposure, but the margin above $\gamma$ remains 
large.

\begin{table}[h]
\centering
\caption{Training dataset ($K{=}10$, $N{=}5$, 
$M{=}5$, $r{=}3$, $T{=}5$, ratio 20\%).}
\label{tab:dataset}
\small
\begin{tabular}{lcccc}
\toprule
Dataset & TPR (\%) & FPR (\%) 
& $\bar{z}_{\rm pos}$ & $\bar{z}_{\rm neg}$ \\
\midrule
UltraChat-200K & 100.0 & 0.0 & 14.12 & 0.57 \\
Alpaca-52K     & 100.0 & 0.0 & 10.23 & 0.32 \\
OpenOrca-100K  & 100.0 & 0.0 & 13.45 & 0.67 \\
\bottomrule
\end{tabular}
\end{table}

\section{Scalability and Overhead}\label{app:overhead}

\subsection{Scalability}
\label{app:scalability}

We scale $K$ from 10 to 100 using Llama-3.2-1B 
with Fictitious Knowledge watermark ($T{=}5$, $M{=}5$, 
$r{=}\lfloor 0.3K \rfloor$) under two subset-size 
strategies. We use the smaller 1B model because 
$K{=}100$ requires training 100 local LoRA 
adapters per round, making the 3B model 
prohibitively expensive on our 4$\times$H200 
cluster. Our model robustness results 
(Table~\ref{tab:model}) suggest that conclusions 
transfer across model scales. 
Table~\ref{tab:scalability} reports the results.

Under fixed $N{=}4$, the signal 
$\bar{z}_{\rm pos}$ decreases from 10.12 to 
7.83 as $K$ grows from 10 to 100, consistent 
with Theorem~\ref{thm:variance}: more non-target 
clients increase the masking noise in the 
estimator. However, $\bar{z}_{\rm pos}$ remains 
well above $\gamma{=}4$ even at $K{=}100$ 
(margin $= 3.83$), and FedAttr maintains 100\% 
TPR and 0\% FPR throughout. The benign signal 
$\bar{z}_{\rm neg}$ increases modestly (0.57 
$\to$ 1.54), but stays far below $\gamma$.

Under proportional $N{=}\lfloor K/3 \rfloor$, 
the privacy--variance trade-off adapts to the 
cohort size: larger $N$ provides stronger 
privacy ($O(d^*/N)$ leakage) while the variance 
factor $N(K{-}1{-}N)/(K{-}2)$ stays controlled. 
At $K{=}100$ with $N{=}33$, the signal 
$\bar{z}_{\rm pos}{=}7.92$ is comparable to the 
fixed-$N$ setting (7.83), while the per-round 
MI leakage is reduced by a factor of 
$33/4 \approx 8\times$.

\begin{table}[t]
\centering
\caption{Scalability with increasing $K$ 
(Llama-3.2-1B, Fictitious Knowledge watermark, FedIT, $T{=}5$, 
$M{=}5$, $r{=}\lfloor 0.3K \rfloor$). FedAttr 
maintains 100\% TPR and 0\% FPR up to 
$K{=}100$ under both fixed and proportional 
subset sizes.}
\label{tab:scalability}
\small
\begin{tabular}{lccccc}
\toprule
$K$ & $N$ & $r$ & TPR & FPR 
& $\bar{z}_{\rm pos}$ / $\bar{z}_{\rm neg}$ \\
\midrule
\multicolumn{6}{l}{\textit{Fixed subset size 
($N{=}4$)}} \\
10  & 4  & 3  & 100\% & 0\% & 10.12 / 0.43 \\
20  & 4  & 6  & 100\% & 0\% & 9.40 / 1.02 \\
50  & 4  & 15 & 100\% & 0\% & 7.86 / 1.32 \\
100 & 4  & 30 & 100\% & 0\% & 7.83 / 1.54 \\
\midrule
\multicolumn{6}{l}{\textit{Proportional 
($N{=}\lfloor K/3 \rfloor$)}} \\
10  & 3  & 3  & 100\% & 0\% & 10.41 / 0.45 \\
20  & 6  & 6  & 100\% & 0\% & 9.56 / 0.56 \\
50  & 16 & 15 & 100\% & 0\% & 7.55 / 1.02 \\
100 & 33 & 30 & 100\% & 0\% & 7.92 / 1.03 \\
\bottomrule
\end{tabular}
\end{table}

\subsection{Overhead}
\label{app:overhead-detail}

Table~\ref{tab:overhead} breaks down FedAttr's 
computational overhead. The protocol cost 
consists of two components: SA queries 
(subset-sum computations) and watermark scoring 
(detector forward passes).

\begin{table}[h]
\centering
\caption{Overhead breakdown (Llama-3.2-3B, 
LoRA rank 64, 4$\times$H200 GPUs).}
\label{tab:overhead}
\small
\begin{tabular}{lrr}
\toprule
Component & Time & \% of training \\
\midrule
FL training (5 rounds) & 8.5 hr & --- \\
\midrule
SA queries (500 total) & 5 min & 1.0\% \\
\quad per query & 0.6 s & \\
FF scoring (55 evals) & 27 min & 5.3\% \\
\quad per eval & $\sim$30 s & \\
\midrule
\textbf{FedAttr total} & \textbf{32 min} 
& \textbf{6.3\%} \\
\bottomrule
\end{tabular}
\end{table}

\paragraph{SA query scaling.}
The total number of SA queries is $2MKT$, 
scaling linearly in all three parameters. 
Table~\ref{tab:query-scaling} reports the 
query count and estimated time for different 
configurations.

\begin{table}[h]
\centering
\caption{SA query count and estimated time 
($M{=}5$, $T{=}5$, 0.6s per query).}
\label{tab:query-scaling}
\small
\begin{tabular}{cccc}
\toprule
$K$ & Queries ($2MKT$) & Time & \% of training \\
\midrule
10 &  500 &  5 min & 1.0\% \\
20 & 1000 & 10 min & 2.0\% \\
50 & 2500 & 25 min & 4.9\% \\
\bottomrule
\end{tabular}
\end{table}

\paragraph{Scoring cost.}
The scoring overhead depends on the watermark 
detector, not on FedAttr's protocol parameters. 
Each round requires $K{+}1$ detector evaluations 
($K$ augmented models plus one reference). Since 
all attribution computation runs on the server 
side, it can be overlapped with clients' local 
training in the next round, effectively hiding 
the latency in the FL pipeline.

\paragraph{Rejection rate.}
The rejection check 
(Section~\ref{sec:surrogate}) ensures that the 
privacy bound (Theorem~\ref{thm:diffuse-query}) 
holds pointwise for every accepted query design. 
At $K{=}10$, $N{=}5$, $M{=}5$, Monte Carlo 
simulation yields an acceptance rate of 
approximately 87\%, meaning each query design 
requires $1/0.87 \approx 1.15$ sampling attempts 
on average. Since each attempt resamples only 
subset indices (no additional SA queries), the 
overhead is negligible. At $K{=}50$, the 
acceptance rate exceeds 99.8\%. Without the 
rejection check, the privacy bound still holds 
in expectation over the query design, matching 
the guarantee of 
\citet{elkordy2023howmuch}.

\subsection{Partial Participation}
\label{app:partial}

In practice, not all clients may be available 
every round. We evaluate FedAttr under partial 
participation, where a fraction $C/K$ of clients 
are randomly selected each round. For client $i$ 
participating in rounds 
$\mathcal{T}_i \subseteq [T]$, the Stouffer 
statistic uses only those rounds:
\begin{equation}
Z_i = \frac{1}{\sqrt{|\mathcal{T}_i|}} 
\sum_{t \in \mathcal{T}_i} z_i^{(t)}.
\end{equation}
Theorem~\ref{thm:stouffer} applies with $T$ 
replaced by $|\mathcal{T}_i|$: fewer rounds 
weaken the signal but the error bound retains 
the same exponential form.

We experiment with $K{=}20$, $r{=}6$, $N{=}5$, 
$M{=}5$, $T{=}10$, and participation rates 
$C/K \in \{0.5, 0.7, 1.0\}$. 
Table~\ref{tab:partial} reports the results. 
FedAttr maintains 100\% TPR and 0\% FPR at all 
three rates. The signal $\bar{z}_{\rm pos}$ 
decreases from 14.5 (full participation) to 
10.8 ($C/K{=}0.5$), consistent with the 
$\sqrt{|\mathcal{T}_i|}$ scaling: at 
$C/K{=}0.5$, each client participates in 
$\approx 5$ rounds, and 
$14.5 \times \sqrt{5/10} \approx 10.3$, close 
to the observed 10.8. The benign signal 
$\bar{z}_{\rm neg}$ increases modestly 
(0.67 $\to$ 1.34), reflecting the reduced 
averaging effect, but remains far below 
$\gamma{=}4$.

These results confirm that FedAttr naturally 
accommodates partial participation: the Stouffer 
aggregation adapts to each client's participation 
history, and attribution remains reliable as 
long as each client accumulates sufficient 
rounds. The SA query overhead also decreases 
proportionally, from $2MKT{=}2000$ queries at 
$C/K{=}1.0$ to $2MCT{=}1000$ at $C/K{=}0.5$.

\begin{table}[h]
\centering
\caption{Partial participation ($K{=}20$, 
$r{=}6$, $N{=}5$, $M{=}5$, $T{=}10$, 
Llama-3.2-1B, Fictitious Knowledge watermark). $C/K$: fraction 
of clients participating per round; 
$\mathbb{E}[|\mathcal{T}_i|]$: expected number 
of rounds per client.}
\label{tab:partial}
\small
\begin{tabular}{ccccc}
\toprule
$C/K$ & $\mathbb{E}[|\mathcal{T}_i|]$ 
& TPR (\%) & FPR (\%) 
& $\bar{z}_{\rm pos}$ / $\bar{z}_{\rm neg}$ \\
\midrule
1.0 & 10 & 100.0 & 0.0 & 14.5 / 0.67 \\
0.7 &  7 & 100.0    & 0.0  & 12.8 / 1.23 \\
0.5 &  5 & 100.0    & 0.0 & 10.8 / 1.34 \\
\bottomrule
\end{tabular}
\end{table}


\paragraph{Discussion.}
At $C/K{=}0.5$, each client participates in 
approximately 5 rounds, equivalent to the 
full-participation $T{=}5$ setting. The signal 
strength should therefore be comparable. At 
$C/K{=}0.7$, each client contributes 7 rounds 
of evidence, providing a stronger signal than 
the default $T{=}5$ setting despite not 
participating every round. This demonstrates 
that FedAttr naturally accommodates partial 
participation: the Stouffer aggregation 
automatically adapts to each client's 
participation history, and the theoretical 
guarantees extend with $|\mathcal{T}_i|$ 
replacing $T$.

The SA query overhead under partial 
participation is $2MC \cdot T$ (only participating clients are queried), 
reducing the total overhead proportionally.

\begin{table}[h]
\centering
\caption{Recovery under non-IID ($\alpha{=}0.1$) with increasing $T$
(K{=}10, N{=}5, M{=}5, r{=}3, $\gamma${=}4.0).}
\label{tab:noniid-recovery}
\begin{tabular}{c cccc}
\toprule
$T$ & TPR (\%)$\uparrow$ & FPR (\%)$\downarrow$ & $\bar{z}_{\text{pos}}$ & $\bar{z}_{\text{neg}}$ \\
\midrule
5  & 67.0  & 0.0  & 8.57  & 0.08 \\
10 & 100.0    & 0.0   & 12.9  & 0.54 \\
15 & 100.0    & 0.0   & 15.8  & 1.02 \\
20 & 100.0    & 0.0   & 18.4  & 1.32 \\
\bottomrule
\end{tabular}
\end{table}

\subsection{Non-IID Recovery via Increasing Communication Rounds}
\label{app:noniid-recovery}

Fig.~\ref{fig:robustness} shows that under severe non-IID heterogeneity ($\alpha{=}0.1$), FedAttr's TPR degrades to 67\% at $T{=}5$.
Theorem~\ref{thm:stouffer} predicts that the Stouffer statistic grows as $\sqrt{T}$, so increasing communication rounds should recover attribution accuracy.
We verify this by sweeping $T \in \{5, 10, 15, 20\}$ at $\alpha{=}0.1$, with all other parameters held at defaults.

Table~\ref{tab:noniid-recovery} confirms the theoretical prediction.
Doubling the rounds from $T{=}5$ to $T{=}10$ restores 100\% TPR and 0\% FPR, with $\bar{z}_{\text{pos}}$ increasing from 8.57 to 12.9.
The observed growth ratio $12.9 / 8.57 \approx 1.51$ is close to the theoretical $\sqrt{10/5} \approx 1.41$, consistent with the $\sqrt{T}$ scaling predicted by Theorem~3.
Further increasing $T$ to 15 and 20 continues to widen the margin ($\bar{z}_{\text{pos}} - \gamma$ reaches 14.4 at $T{=}20$).
The benign signal $\bar{z}_{\text{neg}}$ increases modestly (0.08 to 1.32) but remains well below $\gamma{=}4$ throughout, confirming that additional rounds selectively amplify the watermark signal without inflating false positives.

These results suggest a practical guideline for non-IID deployments: the server can monitor the Stouffer margin across rounds and continue attribution until a target confidence level is reached, rather than fixing $T$ in advance.

\section{Limitations}

\noindent\textbf{Corpus owner involvement.}
FedAttr requires the corpus owner to participate in the attribution phase by evaluating the scoring function with its private detection key.
This introduces an operational dependency: attribution cannot proceed without the corpus owner's cooperation.
In settings where the corpus owner is unavailable or unwilling to participate, a delegated or threshold-based key-sharing mechanism would be needed, which we leave to future work.
\section{Broader Impacts}
\label{sec:broader-impacts}
 
FedAttr is designed to enforce data-use license compliance in federated learning, helping corpus owners identify unauthorized use of their intellectual property without compromising the privacy of honest participants.
By operating entirely within the secure aggregation framework, FedAttr preserves the core privacy guarantees that make federated learning attractive for privacy-sensitive domains such as healthcare and finance.
 
On the positive side, FedAttr strengthens the trust ecosystem between data providers and model trainers: corpus owners gain a practical enforcement tool, which in turn may encourage broader data sharing under clear licensing terms and foster more open collaboration in federated settings.
 
A potential concern is that the attribution mechanism could be repurposed beyond its intended license-enforcement scope---for instance, to monitor or profile individual clients' training behavior.
We note that FedAttr's design mitigates this risk in two ways: (i) the corpus owner must hold the watermark detection key to produce attribution decisions, so the server alone cannot perform attribution; and (ii) the mutual-information bound (Theorem~4) formally limits what the released estimator reveals about any individual client's update.
Nonetheless, deployment guidelines should clearly specify the permissible scope of attribution queries to prevent misuse.

\end{document}